\documentclass[9pt,shortpaper,twoside,web]{ieeecolor}
\usepackage{generic}
\usepackage{cite}
\usepackage{amsmath,amssymb,amsfonts}
\usepackage[skip=5pt]{caption}
\usepackage{graphicx}
\usepackage{textcomp}
\def\BibTeX{{\rm B\kern-.05em{\sc i\kern-.025em b}\kern-.08em
		T\kern-.1667em\lower.7ex\hbox{E}\kern-.125emX}}

\usepackage{cite}
\usepackage[hidelinks]{hyperref}
\usepackage{textcomp}
\usepackage{algpseudocode,algorithm}
\usepackage{float}
\usepackage{stfloats}
\usepackage{textcomp}
\usepackage{tabularx}
\usepackage{bm}
\usepackage{multirow}
\usepackage{pgfplots}
\usepackage{todonotes}
\usepackage{tikz}
\usetikzlibrary{decorations.pathreplacing,calligraphy}
\usetikzlibrary{calc,matrix,arrows.meta}
\pgfplotsset{compat=1.18} 
\usetikzlibrary{arrows}
\usetikzlibrary{arrows.meta}

\newcommand{\Eb}{{\mathbb{E}}}
\newcommand{\Rb}{{\mathbb{R}}}
\newcommand{\Cb}{{\mathbb{C}}}
\newcommand{\Ib}{{\mathbb{I}}}
\newcommand{\Zb}{{\mathbb{Z}}}

\newcommand{\Ec}{{\mathcal{E}}} 

\newcommand{\Cc}{{\mathcal{C}}}
\newcommand{\Ic}{{\mathcal{I}}}

\newcommand{\Uc}{{\mathcal{U}}}
\newcommand{\Tc}{{\mathcal{T}}}
\newcommand{\Sc}{{\mathcal{S}}}

\newcommand{\Mc}{{\mathcal{M}}}
\newcommand{\Nc}{{\mathcal{N}}}

\newcommand{\Jc}{{\mathcal{J}}}

\newcommand{\be}{\mathbf{e}}

\newcommand{\m}{\text{-}}

\newcommand{\ift}{{\infty}}
\newcommand{\rs}{\text{rowspan}}

\newcommand{\diag}{{\mathrm{diag}}}

\newcommand{\norm}[1]{\left\lVert#1\right\rVert}

\newcommand{\ls}{\text{ls}}

\newcommand{\ba}{\left[ \begin{array}}
\newcommand{\ea}{\\ \end{array} \right]}

\newcommand*{\QET}{\hfill\ensuremath{\triangleleft}}

\DeclareMathOperator{\cov}{Cov}

\DeclareMathOperator{\spe}{sp}

\newtheorem{theorem}{Theorem}
\newtheorem{proposition}{Proposition}
\newtheorem{lemma}{Lemma}

\newtheorem{definition}{Definition}

\newtheorem{assumption}{Assumption}
\newtheorem{remark}{Remark}

\makeatletter
\newcommand{\multiline}[1]{%
  \begin{tabularx}{\linewidth-\ALG@thistlm-0.0cm}[t]{@{}X@{}}
    #1
  \end{tabularx}
}

\begin{document}
	\title{Secure Filtering against Spatio-Temporal False Data Attacks under Asynchronous Sampling}
	\author{Zishuo Li$^*$, Anh Tung Nguyen$^*$, Andr\'e M. H. Teixeira, Yilin Mo, Karl H. Johansson
    \thanks{$^*$ Zishuo Li and Anh Tung Nguyen contributed equally to this work.}
	\thanks{This work is supported by the National Natural Science Foundation of China under grant no. 62273196, the Swedish Research Council under the grant 2021-06316, 2023-05234, the Swedish Foundation for Strategic Research, the Swedish Research Council Distinguished Professor grant 2017-01078, and the Knut and Alice Wallenberg Foundation Wallenberg Scholar grant.
    }  
    \thanks{Zishuo Li and Yilin Mo are with the Department of Automation, Tsinghua University, Beijing, 100084,
		China. {\tt\small \{lizs19@mails, ylmo@mail\}.tsinghua.edu.cn}.}
	\thanks{Anh Tung Nguyen and Andr{\'e} M. H. Teixeira are with the Department of Information Technology, Uppsala University, PO Box 337, SE-75105, Uppsala, Sweden. {\tt\small \{anh.tung.nguyen, andre.teixeira\}@it.uu.se}.}
    \thanks{Karl H. Johansson is  with School of
		Electrical Engineering and Computer Science, KTH Royal Institute of
		Technology, Stockholm SE-100 44, Sweden. He is also affiliated with Digital Futures. {\tt\small kallej@kth.se}.}
    }
\maketitle


%

\begin{abstract}                          
	

    This paper addresses the secure state estimation problem for continuous linear time-invariant systems with non-periodic and asynchronous sampled measurements, where the sensors need to transmit not only measurements but also sampling time-stamps to the fusion center. 
    This measurement and communication setup is well-suited for operating large-scale control systems and, at the same time, introduces new vulnerabilities that can be exploited by adversaries through
    (i) manipulation of measurements, (ii) manipulation of time-stamps, (iii) elimination of measurements, (iv) generation of completely new false measurements, or a combination of these attacks.
    To mitigate these attacks, we propose a decentralized estimation algorithm in which each sensor maintains its local state estimate asynchronously based on its measurements. 
    The local {estimates} are synchronized through time prediction and fused after time-stamp alignment.
    In the absence of attacks, state estimates are proven to recover the optimal Kalman estimates by solving a weighted least squares problem.
    In the presence of attacks, solving this weighted least squares problem with $\ell_1$ regularization provides secure state estimates with uniformly bounded error under an observability redundancy assumption.
    The effectiveness of the proposed algorithm is demonstrated using a benchmark example of the IEEE 14-bus system.
\end{abstract}


\section{Introduction}
Many real-world large-scale systems, such as power systems, water distribution networks, and transportation networks, are examples of cyber-physical systems where physical plants are tightly coupled with digital devices. 
These systems are monitored and controlled via wired or wireless communications, leaving them vulnerable to malicious attackers.
A recent report has shown the disastrous consequences of 
malware for industrial control systems 
\cite{miller2021looking}.
The challenge of securely estimating states under malicious activities has been widely addressed \cite{liu2020local,li2023efficient,shoukry2015event,nakahira2018attack,hu_stateestimation_automatica,lu2019resilient,liu2021resilient}, given their crucial role in control systems.
This paper contributes to
secure state estimation by considering asynchronous and non-periodic measurements under false data attacks.


To deal with the problem of secure state estimation against false data injection attacks, three research directions consisting of the sliding window method, the estimator switching method, and the local decomposition-fusion method, have been developed in recent years \cite{shoukry2015event,nakahira2018attack,liu2020local,li2023efficient}.
The sliding window method considers past sensor measurements in a finite-time horizon to provide state estimates via batch optimization 
\cite{shoukry2015event}.
The estimator switching method maintains multiple parallel estimators, each utilizing measurements from a subset of all sensors 
\cite{nakahira2018attack}.
The local decomposition-fusion method deploys multiple decentralized estimators, each of which samples the measurement of one local sensor. 
Then, local state estimates are fused by a convex optimization to generate secure state estimates \cite{liu2020local,li2023efficient}.

Denial-of-service (DoS) attacks block the measurement transmitted to the state estimator, undoubtedly worsening the state estimation performance. To handle such attacks conducted on multiple transmission channels, a partial observer is proposed to provide reliable partial state estimates in \cite{lu2019resilient}. 
The authors in \cite{su2018cooperative} propose a detection-compensation scheme to detect the presence of DoS attacks and then effectively reconstruct missing state estimates through past available states, eventually mitigating the attack's impact on the state estimation performance.
Event-triggered mechanisms are proven to be effective against DoS attacks in synchronous sampled systems such as the dynamic event-triggered scheme proposed in \cite{liu2021resilient}.

\begin{figure}[!t]
	\centering
	\begin{tikzpicture}[scale=0.86]
%
%
%
%
%
%

\draw[-latex][black] (0, 0)--(3.5,0);
\draw[-latex][black] (0, 0)--(0,2.75);
\node [rotate=90] at (-0.25,1.5) {\scriptsize state};
\node at (3.25,-0.25) {\scriptsize time axis};

\draw[-latex][black] (4,0)--(7.5,0);
\draw[-latex][black] (4,0)--(4,2.75);
\node [rotate=90] at (3.75,1.5) {\scriptsize state};
\node at (7.25,-0.25) {\scriptsize time axis};

\draw[-latex][black] (0,-3.75)--(3.5,-3.75);
\draw[-latex][black] (0,-3.75)--(0,-1);
\node [rotate=90] at (-0.25,-2.25) {\scriptsize state};
\node at (3.25,-4) {\scriptsize time axis};

\draw[-latex][black] (4,-3.75)--(7.5,-3.75);
\draw[-latex][black] (4,-3.75)--(4,-1);
\node [rotate=90] at (3.75,-2.25) {\scriptsize state};
\node at (7.25,-4) {\scriptsize time axis};

\filldraw (1.25,0.75) circle (2pt);
\filldraw [red] (1.25,2.25) circle (2pt);

\filldraw (0.5,1.5)  circle (2pt);
\filldraw (2,2)  circle (2pt);
\filldraw (2.5,1.25)  circle (2pt);
\filldraw (3,1.5)   circle (2pt);

   \draw[-{Stealth[length=3mm, width=2mm]}][red](1.25,0.75)--(1.25,2.15);
\node at (1.75,3) {a) false-data injection};

\draw[thick] plot [smooth,tension=1.5] coordinates{(0,2.5) (0.5,1.5) (1.25,0.75) (2,2) (2.5,1.25) (3,1.5) (3.25,1.5)};

\filldraw (5.25,0.75) circle (2pt);
\filldraw [red] (6.75,0.75) circle (2pt);

\filldraw (4.5,1.5)  circle (2pt);
\filldraw (6,2)  circle (2pt);
\filldraw (6.5,1.25)  circle (2pt);
\filldraw (7,1.5)   circle (2pt);

   \draw[-{Stealth[length=3mm, width=2mm]}][red] (5.25,0.75)--(6.7,0.75);
\node at (5.75,3) {b) time-stamp manipulation};
\draw[thick] plot [smooth,tension=1.5] coordinates{(4,2.5) (4.5,1.5) (5.25,0.75) (6,2) (6.5,1.25) (7,1.5)  (7.25,1.5)};

\filldraw (1.25,-3) circle (2pt);

\filldraw (0.5,-2.25)  circle (2pt);
\filldraw (2,-1.75)  circle (2pt);
\filldraw (2.5,-2.5)  circle (2pt);
\filldraw (3,-2.25)   circle (2pt);

\node at (1.75,-0.75) {c) denial-of-service};
\draw [red][line width=1.2pt](1,-3.25)--(1.5,-2.75);
\draw [red][line width=1.2pt](1.5,-3.25)--(1,-2.75);

\draw[thick] plot [smooth,tension=1.5] coordinates{(0,-1.25) (0.5,-2.25) (1.25,-3) (2,-1.75) (2.5,-2.5) (3,-2.25)  (3.25,-2.25)};

\filldraw (1.25,-3) circle (2pt);
\filldraw [red] (6.75,-1.75) circle (2pt);

\filldraw (4.5,-2.25)  circle (2pt);
\filldraw (5.25,-3)  circle (2pt);
\filldraw (6,-1.75)  circle (2pt);
\filldraw (6.5,-2.5)  circle (2pt);
\filldraw (7,-2.25)   circle (2pt);

\node at (5.75,-0.75) {d) false-data generation};

\draw[thick] plot [smooth,tension=1.5] coordinates{(4,-1.25) (4.5,-2.25) (5.25,-3) (6,-1.75) (6.5,-2.5) (7,-2.25)  (7.25,-2.25)};


\filldraw (1.5,-4.5)  circle (2pt); 
\node at (4.,-4.5) {\small authentic measurement};

\filldraw (1.5,-5)  circle (2pt); 
\draw [red][line width=1.2pt](1.25,-4.75)--(1.75,-5.25);
\draw [red][line width=1.2pt](1.75,-4.75)--(1.25,-5.25);
\node at (4.,-5) {\small blocked measurement};

\filldraw [red] (1.5,-5.5)  circle (2pt); 
\node at (4.,-5.5) {\small false measurement};




\end{tikzpicture}
	\caption{Examples of spatio-temporal false data attacks that can manipulate both time-stamps and measurements.} \label{fig:time_data_attack}
    \vspace{-0.5cm}
\end{figure}

The asynchronous and non-periodic sampling scheme opens up new opportunities for the adversaries. 
In this paper, we propose a novel false data attack model for such systems, which was also introduced in the preliminary version \cite{li2023secure} (see Fig.~ \ref{fig:time_data_attack}). This attack model includes both integrity attacks such as false-data injection \cite{hu_stateestimation_automatica}, and availability attacks such as DoS attacks \cite{lu2019resilient,liu2021resilient}. We also investigate the influence of time-stamp manipulation, {which can significantly degrade state-estimation performance \cite{qu2023time,zhang_TSA_TSG}.}
Apart from these attacks, generating completely new false data packages into authentic measurement streams is also a serious threat.
Our attack model unifies all the above attacks into one framework including the possibility of their combinations. {It is worth noting that the preliminary version \cite{li2023secure} proposed a median-based state estimation, which is resilient to attacks at the cost of conservatism and accuracy. These limitations are addressed by introducing a completely distinct state estimation algorithm in this paper.}



To the best of our knowledge, little progress has been made toward studying time-stamp manipulation on state estimation performance, especially on asynchronous sampled systems.
Li et al. \cite{LI2008199} 
propose Kalman filter (KF) algorithms for non-uniformly sampled multi-rate systems. To deal with the problem of asynchronous sampled systems, the authors in
\cite{Feddaoui_asy_kalman} propose an observer for continuous-time systems with discrete measurements, resulting in a differential  Riccati equation. 
Ding et al.\cite{DING2009324} 
analyze the observability degradation problem of multi-rate and non-periodic sampled systems.

The main contribution of the paper is a secure estimation algorithm that recovers the system state in the presence of spatio-temporal false data attacks. The algorithm has the following merits:
\begin{enumerate}
    \item The algorithm is built upon the decomposition of KF, which provides local state estimates. We first introduce a weighted least squares optimization-based fusion of local state estimates. We show that the result of the fusion is exactly the optimal state estimates obtained by KF in the absence of attacks.
    \item To enhance security against attacks, we improve the weighted least squares optimization-based fusion by adding an $\ell_1$ regularization. In the absence of attacks, we provide a sufficient condition on design parameters under which state estimates provided by the $\ell_1$-regularization fusion and KF are identical, and thus optimal. 
    \item In the presence of attacks, the $\ell_1$-regularization fusion provides a secure state estimate whose estimation error is independent of attacks and is directly related to
    the estimation error of an oracle KF operating in the attack-free scenario. This merit highlights the effectiveness of the $\ell_1$-regularization fusion in mitigating the impact of attacks. 
\end{enumerate}

The effectiveness of the secure state estimation algorithm is validated through the IEEE 14-bus system.
We conclude this section by introducing the notation that will be utilized throughout this paper.

\textit{{Notation}:}
The sets of positive integers, non-negative integers, and non-negative real numbers are denoted as $\Zb_{>0},\Zb_{\geq0}$, and $\Rb_{\geq0}$, respectively. 
For a real number $x$, $\lceil x \rceil$ represents $x$ rounded up to the nearest integer.
The cardinality of a set $\Sc$ is denoted as $|\Sc|$.  Denote the span of row vectors of matrix $A$ as $\rs(A)$. 
We denote $I$ as an identity matrix with an appropriate dimension.
The spectrum of matrix $A$ is denoted as $\spe(A)$.
For a vector $x$, $[x]_j$ stands for its $j$-th entry. 
We denote the continuous time index in a pair of parentheses $(\cdot)$ and the discrete-time index in a pair of brackets $[\cdot]$. Let
$\partial f(x)$ be the subgradient of function $f$ at $x$.

\section{Problem Formulation}
\label{sec:problem}	
We first introduce the system, the modeling of asynchronous measurements, and several assumptions that will be used throughout this paper. Secondly, we present a general notion of spatio-temporal false data attacks. Finally, the secure estimation problem is formulated.
\subsection{Systems with asynchronous measurements}
\label{sec:basic_measure}
Throughout this paper, we consider a continuous linear time-invariant (LTI) system mathematically described by $n$ states and measured by $m$ sensors. Let us denote the state index set as $\Jc \triangleq\{1,2, \ldots, n\}$ and the sensor index set as $\Ic \triangleq\{1,2, \ldots, m\}$. The LTI system is modeled as follows:
\begin{align}
	\dot{x}(t)&=A x(t)+w(t), \label{eq:system} \\
	y_i(t) &= C_i x(t) + v_i(t), ~ \forall \, i \in \Ic, \label{eq:y_i_def}
\end{align}
where $x(t) \in \mathbb{R}^{n}$ is the system state. The process noise 
$w(t) \in \mathbb{R}^{n}$ is continuous zero-mean Gaussian noise with the power spectral density ${w}(t) \sim \Nc(0,Q)$, where $Q$ is a given positive definite matrix.
The measurement given by sensor $i$ is denoted by $y_i(t) \in \mathbb{R}$.
Let us denote the measurement matrix of all the sensors $C \triangleq [C_1^\top,\dots, C_m^\top]^\top$ where $C_i^\top \in \Rb^{n}$ is given for all $i \in \Ic$.
The measurement noise vector $v(t) \triangleq [v_1(t),\dots, v_m(t)]^\top$ is zero-mean Gaussian noise with the power spectral density ${v}(t) \sim \Nc(0,R(t))$.
The initial state
$x(0)$ is assumed to be a Gaussian random vector with a known covariance 
and is independent of measurement noises, i.e., $x(0) \sim \Nc(0,\,\Sigma)$ where $\Sigma$ is known.
Let us introduce the following assumption. 
\begin{assumption}\label{as:R}
		The measurement noise covariance $R(t)$ is uniformly upper bounded: $0 \preceq R(t) \preceq \bar{R},~\forall~ t\,\in\,\Rb_{\geq0}$,
		where $\bar{R}$ is a given constant positive definite matrix.
  \QET
\end{assumption}

The sensors sample and send data packets to an estimator in a non-periodic and asynchronous manner, which contain not only measurements but also their sensor indices and sampling time-stamps. More specifically, the estimator receives measurement triples from sensor $i \in \Ic$, which has the following form:
\begin{align}
\label{measurement3}
	\textbf{measurement triple: }(i, \, t,\, y_i(t)),
\end{align}
where $i$ is the sensor index, $t$ is the sampling time-stamp, and $y_i(t)$ is the measurement given by sensor $i$. 

Define the set of sampling time-stamps from sensor $i$ as $\Gamma_i$. Without loss of generality, the time when the estimator starts working is set as $t_0=0$. 		
In order to guarantee system observability under non-uniform asynchronous measurements, we introduce the following notation.
Define the set of sampling time intervals and cumulative sampling time from sensor $i$ as follows
\begin{align*}
	\Tc\triangleq \textstyle \bigcup_{i=1}^m \Tc_i,\ \Tc_i &\triangleq \left\{t_k-t_{k-1}~ | ~t_k,t_{k-1}\in\Gamma_i,k\in\Zb_{>0} \right\},        \\
	\widetilde{\Tc}\triangleq \textstyle \bigcup_{i=1}^m \widetilde \Tc_i,\ \widetilde{\Tc}_i &\triangleq \left\{t_k-t_{j}~ | ~t_k,t_j\in\Gamma_i, k>j ,k,j\in\Zb_{\geq 0}\right\}    .
\end{align*}
Define the system pathological sampling interval set~\cite{DING2009324} as
\begin{align*}
	\Tc^*\triangleq \left\{ T>0 ~|~ e^{\lambda_i T}=e^{\lambda_j T},~  i\neq j,\lambda_i,\lambda_j\in \spe(A) \subseteq \Cb \right\} .
\end{align*} 
To prevent system observability degradation problems due to discrete-time sampling, the following assumption, which is also seen in \cite{Muhammad_pathological_sample,DING2009324}, is introduced.
\begin{assumption}[non-pathological sampling time]\label{as:sample_time}
    Given a positive number $T_{\max}$,
	the sampling time interval sets $\Tc$ and $\tilde \Tc$ satisfy the following conditions: $\sup \Tc \leq T_{\max}~ \text{and}~ \widetilde{\Tc}\cap \Tc^*=\varnothing$, i.e., the sampling interval set is upper-bounded by $T_{\max}$ and has no intersection with the pathological sampling interval set $\Tc^*$.
 \QET
\end{assumption} 

\subsection{Spatio-temporal false data attacks}
We introduce a new spatio-temporal false data attack that generalizes integrity attacks and availability attacks~(see Fig.~\ref{fig:time_data_attack}). More specifically, the adversary may manipulate the entire measurement triple \eqref{measurement3} rather than only the measurement.
Let us denote $\Sc(t)$ as the set of all authentic measurement triples with time-stamp $t$:
\begin{align*}
	\Sc(t) \triangleq \left\{(i, \, t, \, y_i(t))~|~\forall i\in\Ic  \right\}. 
\end{align*}
Moreover, $\Sc^a(t)$ denotes the set of {authentic and possibly manipulated} measurement triples with time-stamp $t$. 
Denote the set of corrupted sensors as $\Cc$, which is supposed to be fixed over time and unknown to the operator. 
{Studying time-varying $\Cc$ is left for future work.}
Now, we are ready to define the spatio-temporal false data attack as follows:
\begin{figure*}[!t]
    \centering
    \includegraphics[width=0.83\linewidth]{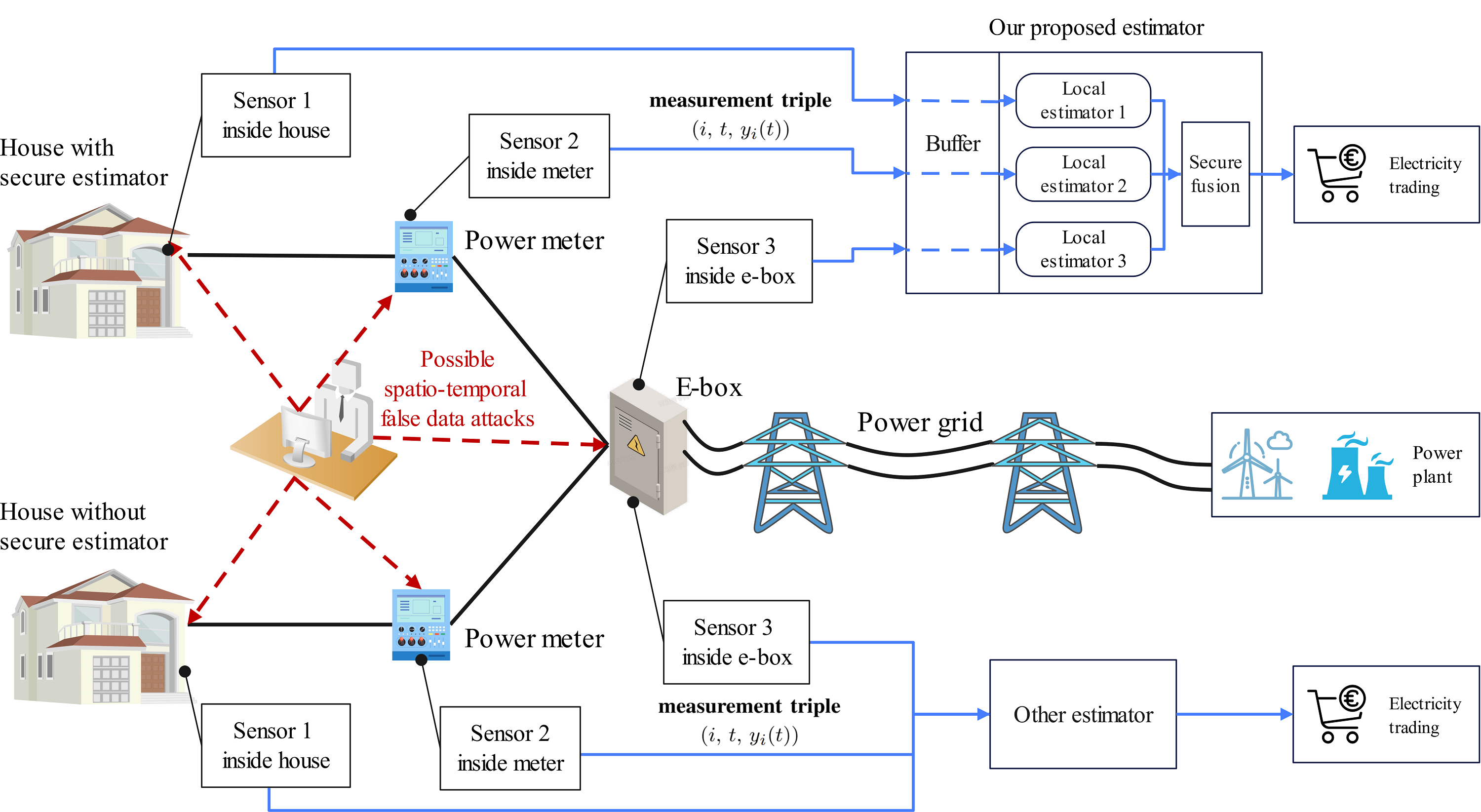}
    \caption{An example of secure state estimation in electricity consumption monitoring. The attacker can launch various types of spatio-temporal false data attacks on sensors.}
    \vspace{-0.6cm}
    \label{fig:power}
\end{figure*}
\begin{definition}[Spatio-temporal false data attacks]\label{def:attack}
	Attackers can manipulate measurement triples given by corrupted sensor $i\in\Cc$ in the following four ways, where $(i, \, t, \, y_i(t)) \in \Sc(t) $ is an authentic measurement triple, $y^a_i(t)$ and $t^a$ are manipulated values from $y_i(t)$ and $t$, respectively, and $( i, \, t^f, \, y^f_i(t))\notin \Sc(t)$ is a newly generated false measurement triple:
    \\
    (1) \textbf{false-data injection:} 
    $S^a(t) \triangleq  \big[ S(t) \setminus (i,  t,  y_i(t)) \big] \bigcup   (i,  t,  y^a_i(t))$, 
    \\
    (2) \textbf{time-stamp manipulation:}  \\
    $~~~~~~~~~~~~~~S^a(t) \triangleq  \big[ S(t) \setminus (i,  t, y_i(t)) \big] \bigcup (i,  t^a, y_i(t))$,
    \\
    (3) \textbf{denial-of-service:} 
	$\Sc^a(t) \triangleq  \Sc(t)\setminus ( i, t, y_i(t))$,
    \\
    (4) \textbf{false-data generation:} 
	$\Sc^a(t) \triangleq \Sc(t) ~\bigcup~ ( i, \, t^f, \, y^f_i(t) )$.
    \\
	Further, if the set of corrupted sensors satisfies $|\Cc|\leq p$, the attack is called $p$-sparse. \QET
\end{definition}
We mainly study $p$-sparse spatio-temporal false data attacks.
Next, we introduce observability redundancy in the following assumption, which is a necessary condition and commonly used in literature (see \cite{fawzi2014secure,lu2019resilient,nakahira2018attack,Mishra2017TCNS,an2017secure}).
\begin{assumption}
    \label{as:2p}
    The system $(A,C)$ is $2p$-sparse observable, i.e., the system $(A,C_{\Ic\setminus\Mc})$ is observable for any subset $\Mc\subset\Ic$ where $|\Mc| = 2p$ and the matrix {$C_{\Ic\setminus\Mc}$} represents the matrix composed of rows of $C$ with row indices in $\Ic\setminus\Mc$. \QET
\end{assumption}

The manipulated time-stamp set $\Gamma^a$ is defined as follows:
\begin{align}
	\Gamma^a\triangleq & \textstyle \bigcup_{i=1}^m \Gamma^a_i,&\Gamma^a_i\triangleq \{t~|~(i, \, t, \, y_i(t))\in\Sc^a(t)\}.
\end{align}

In practice, received measurement time-stamps may not be in increasing order, leading to the out-of-sequence problem \cite{Kaempchen2003DATASS,out-of-sequence-measurements-IVS,FIR_Discretely_Delayed}. 
This problem can be dealt with by using a \emph{buffer} that sorts measurement triples based on their time-stamps in increasing order \cite{Kaempchen2003DATASS,out-of-sequence-measurements-IVS,FIR_Discretely_Delayed}.
We assume a buffer before the secure estimator, {which may be misled by manipulated time-stamps. Thus, we employ the following assumption 
for the received, possibly manipulated, time-stamp.}
\begin{assumption}\label{as:time_in_order}
	The received measurement triples are {sorted} such that their corresponding time-stamps are in 
    {non-decreasing} order, i.e., $\Gamma^a=\{t_0,t_1,t_2,\dots\}$ and $0 = t_0 < t_z \leq t_{z+1},~\forall \, z \in\Zb_{>0}$. \QET
\end{assumption}

{Assumption~4 stipulates that, after buffering, the received measurement triples are presented in non-decreasing order of their possibly manipulated time-stamps. This post-buffer ordering provides a well-defined discrete-time index $[k]$ for the subsequent design and analysis.}
\subsection{Secure state estimation problem}
An example of state estimation problem in electricity monitoring is given in Fig.~\ref{fig:power} where an attacker conducts spatio-temporal false data attacks.
To deal with such attacks,
we design a secure state estimation algorithm that provides a secure state estimate $\check x(t)$ of the true state $x(t)$ with uniformly bounded error:
\begin{align}
        \lvert \, \check x_j[k] - x_j[k] \, \rvert \leq F(A,C,Q,\Sigma,\bar{R},\gamma),\, \forall \, j \, \in \, \Jc,\label{uni_bound}
    \end{align}
where the design parameter $\gamma$ is a positive scalar and $\Sigma = \Eb \, [ \, x(0) x(0)^\top \, ]$. Notice that the value of the function $F(\cdot)$ depends only on the system parameters and remains independent of attacks, {whose explicit form will be provided later in Theorem~\ref{th:secure_fusion}}.  


In the following section, we present the {asynchronous} sampled-data KF, {which is known to provide an optimal state estimate in the absence of attacks \cite{Feddaoui_asy_kalman},}  
and its local linear decomposition. 
{The asynchronous sampled-data KF serves as a reference estimator without the effects of sensor attacks (hereafter, the “oracle KF”), which we use in our analysis to bound the estimation error of the designed estimator.
The local decomposition of the asynchronous sampled-data KF inherently isolates malicious impacts on each sensor, which supports us in building a secure estimation fusion algorithm. More specifically, the algorithm is shown to exactly recover the oracle KF performance in the absence of attacks, while providing a state estimate that remains in a neighborhood of the oracle KF estimate in the presence of spatio-temporal attacks.
The secure state estimation algorithm is developed and analyzed} in Section~\ref{sec:secure_algorithm}.
\section{Asynchronous Sampled-data KF and Its Decomposition}
\label{sec:hybridkalman}
We first introduce the sampled-data KF with asynchronous sampling measurements. The remainder of the section presents the decomposition of the sampled-data KF and how it recovers state estimates provided by the sampled-data KF.
\subsection{Asynchronous sampled-data KF}
For continuous {LTI} systems with synchronous discrete-time measurements, the sampled-data KF provides optimal state estimates by combining continuous-time prediction steps and discrete-time update steps \cite{sarkka2006recursive}.
We define the measurement availability index $\phi_i[k]\in\{0,1\}$ where $\phi_i[k]=1$  if sensor $i$ has a measurement with time-stamp $t_k$ and $\phi_i[k]=0$ otherwise. The notation $[k]$ stands for the discrete-time instant.
Let us define the following matrices: 
\begin{align}
	A[k]& \triangleq \exp((t_{k+1}-t_{k})  A ), \
    C[k] \triangleq \diag (\phi[k]) \, C, 
    \nonumber \\
    Q[k] &\triangleq \int_{t_k}^{t_{k+1}} \exp(\tau A) Q \exp(\tau A^\top) \, \text{d} \tau,
	\nonumber \\
	R[k] &\triangleq \diag (\phi[k]) \, R(t_k)
    \, \diag (\phi[k]), \label{sampled_matrices}
\end{align}
where $\phi[k]\triangleq\left(\phi_1[k],\dots, \phi_m[k]\right)^\top$. 
At each sampling instant $k$, the system can be considered as a discrete time-variant system, on which we implement the following asynchronous sampled-data KF:
\begin{subequations}\label{eq:asy_kalman}
	\begin{align}
		&\hspace{-12pt}\textbf{Prediction steps:} \notag \\
		&\hat{x}_{\m}[k]=A[k-1]\hat{x}[k-1] ,\\
		&P_{\m}[k]=A[k-1]P[k-1]A^{\top}[k-1] + Q[k-1], \\
		&\hspace{-12pt}\textbf{Update steps:} \notag \\
		&K[k]=P_{\m}[k] C^{\top}[k] 
		\big( C[k] P_{\m}[k] C^{\top}[k]+R[k] \big)^{\dagger},
		\label{eq:def_Kk_asy} 	\\
		& P[k] = (I-K[k] C[k]) P_{\m}[k], \label{eq:def_Pt_asy} \\
		& \hat{x}[k]=\hat{x}_{\m}[k] + K[k]  \left(y[k]-C[k] \hat{x}_{\m}[k] \right) , \label{eq:def_xt_asy}
	\end{align}
\end{subequations}
where $y[k]\triangleq y(t_k)$, initial condition $\hat{x}[0]=0, \, P[0]=\Sigma$, and $(\cdot)^\dagger$ stands for the Moore-Penrose inverse.
Notice that when $\phi_i[k]=0$, $C_i[k]=\bm{0}^{\top}$ and thus based on \eqref{eq:def_Kk_asy}, the $i$-th column of the Kalman gain is zero, i.e., $K_i[k]=\mathbf{0}$, resulting in $K[k]C[k]=K[k]C ~\forall k$. 

In the following section, we will decompose the KF \eqref{eq:def_xt_asy} into a linear sum of local state estimates and propose an optimization-based fusion scheme that provides a state estimate exactly the same as the one given by the KF \eqref{eq:def_xt_asy}.
\subsection{Linear decomposition of the sampled-data KF}
Define
\begin{align}\label{eq:def_Pi}
	\Pi[k-1] &\triangleq A[k-1]-K[k]C A[k-1].
\end{align}
The local estimator at sensor $i$ is defined as:
\begin{align}\label{eq:def_zeta}
	\zeta_{i}[k] = \Pi[k-1] \zeta_{i}[k-1] + K_i[k] y_i[k],
\end{align}
which is initialized as $\zeta_{i}[0]= \mathbf{0}$.
From \eqref{eq:def_xt_asy}, \eqref{eq:def_Pi}, and \eqref{eq:def_zeta}, one obtains the following property:
\begin{equation}\label{eq:sum_zeta=xhat}
	\hat{x}[k]= \textstyle \sum_{i=1}^{m} \, \zeta_{i}[k] .
\end{equation}
In the following, we show the relationship between $\zeta_{i}[k]$ and $x(t_k)$, and prove that $\zeta_{i}[k]$ is a stable estimate of $G_i[k]{x}[k]$ where $G_i[k]$ satisfies the following dynamics:
\begin{equation}\label{eq:defG}
	G_i[k] = \Pi[k-1] G_i[k-1] A^{-1}[k-1]+K_i[k]C_i.
\end{equation}
Note that $G_i[k]$ plays a crucial role in designing the secure state estimation algorithm in Section~\ref{sec:secure_algorithm}. Therefore, we analyze its structure and show that $G_i[k]$ has a time-invariant form in the following.

\subsection{Structure of $G_i[k]$}
We need the following assumption to prevent the observability degradation problems.

\begin{assumption}\label{as:geo_mul}
The geometric multiplicity of all the eigenvalues of $A$ is 1. 
\QET
\end{assumption}

Assumption~\ref{as:geo_mul} simplifies the observability structure of system $(A,C)$, which can be seen from Lemma \ref{lm:span} later. The use of Assumption~\ref{as:geo_mul} ensures that the Jordan blocks of $A$ are linearly independent, which enables the definition of state observability, i.e., define $\Ec_j$ as the index set of sensors that can observe state $j$, i.e.
\begin{equation}\label{eq:def_Oc}
	\Ec_j\triangleq \{ i\in\Ic\ |\ O_i^{\top} \be_j\neq \mathbf{0} \},
\end{equation}
where $\be_j$ is the canonical basis vector with 1 on the $j$-th entry and 0 on the other entries. Moreover, 
$$
O_{i} \triangleq \big[ C_i^\top, \, 
(C_iA)^\top, \, \dots, \, (C_{i} A^{n-1})^\top
\big]^\top
$$
is the observability matrix of the system $(A,C_i)$.
Since we focus on the observable system, the state observability index set $\Ec_j$ is not empty, i.e., $\Ec_j\neq \varnothing,\forall j\in\Jc$.
With Assumption \ref{as:geo_mul}, the following results, {whose proofs are reported in the full version of this paper \cite[Appendices A.5 \& A.6]{li2025secure},} characterize the structure of $G_i[k]$. 
\begin{lemma}\label{lm:span} 
	Given 
    the dynamics \eqref{eq:defG}, if $\rs(G_i[0])=\rs(O_i)$, the following holds $\forall k\in\Zb_{\geq 0}$:
	\begin{align}
		\rs(G_i[k])=\rs(O_i)=\rs(H_i), \label{rowspan_Gi}
	\end{align}
	where $H_i \triangleq \diag \left( \Ib_{\Ec_1}(i), \, \Ib_{\Ec_2}(i), \dots, \, \Ib_{\Ec_n}(i) \right)$ 
	and $\mathbb{I}_{\Ec_j}(i)$ is the indicator function that takes the value 1 when $i \in \Ec_j$ and value 0 when $i \notin \Ec_j$.
	As a result, there exists an invertible matrix $V_i[k]$ such that $V_i[k] G_i[k]=H_i$. 
    \QET
\end{lemma}

\begin{lemma}\label{lm:sumG}
    Given the dynamics \eqref{eq:def_Pi} and \eqref{eq:defG},
	if $\sum_{i=1}^{m} G_i[0]=I$, the following holds for all $k\in\Zb_{\geq 0}$: $\sum_{i=1}^{m} G_i[k]=I$. \QET
\end{lemma}

\begin{remark}
    To fulfill the initialization requirements in Lemmas~\ref{lm:span}--\ref{lm:sumG}, i.e., $\sum_{i=1}^{m} G_i[0]=I$ and $\rs(G_i[0])=\rs(O_i)$, 
    we can initialize the sequence $G_i[k]$ as $G_i[0] = \diag \big( \Ib_{\Ec_1}(i)/|\Ec_1|, \, \Ib_{\Ec_2}(i)/|\Ec_2|, \ldots, \, \Ib_{\Ec_n}(i)/|\Ec_n|\big)$, which is well-defined thanks to observable systems. 
    \QET
\end{remark}

\vspace{-0.1cm}
\subsection{Least-squares state estimation fusion}
Define the local residue as $\epsilon_i[k]\triangleq\zeta_{i}[k]-G_i[k] x[k]$ and the global residue as $
\epsilon[k] \triangleq [\epsilon_1[k]^{\top}, \ldots, \epsilon_m[k]^{\top}]^{\top}$, which has dynamics and covariance matrix {as presented} in the following lemma. 
\begin{lemma}\label{lm:epsilon}
	For a fixed sensor $i$, the local residue $\epsilon_i[k]$ satisfies the following dynamics:
	\begin{align}\label{eq:epsilon_recursive}
		\epsilon_i[k+1]\!=\!\Pi[k]\epsilon_i[k]\!-\!\Pi[k] G_i[k] A^{-1}[k] w[k] 	\!+\!K_i[k+1] v_{i}[k+1].
	\end{align}
Moreover, the covariance matrix of $\epsilon[k]$ is computed as follows: 
\begin{equation}
	\cov(\epsilon[k+1])=\bm{\Pi}[k] \cov(\epsilon[k]) \bm{\Pi}^{\top}[k] +\bm{Q}[k],
	\label{def_cov_epsilon}
\end{equation}
where $\bm{\Pi}[k]\triangleq I_m \otimes \Pi[k]$ and
\begin{align*}
	&\bm{Q}[k]\triangleq \cov\big(\Pi[k] G_i[k] A^{-1}[k] w[k]-K_i[k+1] v_{i}[k+1]\big) \\
	&=\begin{bmatrix}
		\Pi[k]G_1[k]A^{-1}[k] \\
		\vdots \\
		\Pi[k]G_m[k]A^{-1}[k]
	\end{bmatrix}
	Q[k]
	\begin{bmatrix}
		\Pi[k] G_1[k]A^{-1}[k] \\
		\vdots \\
		\Pi[k] G_m[k]A^{-1}[k]
	\end{bmatrix}^{\top}\\
	&\quad +\begin{bmatrix}
		K_1[k+1] \\
		\vdots \\
		K_m[k+1]
	\end{bmatrix}
	\begin{bmatrix}
		K_1[k+1] \\
		\vdots \\
		K_m[k+1]
	\end{bmatrix}^{\top}
	\circ
	\left( R[k+1]\otimes \mathbf{1}_{n\times n} \right),
\end{align*}
the notation $\circ$ denotes the element-wise matrix multiplication, and $\otimes$ denotes the Kronecker product.
 \QET
\end{lemma}
The proof is {reported} in \cite[Appendix A.7]{li2025secure}.
The result of Lemma~\ref{lm:epsilon} shows that the $\zeta_i[k]$ is the stable estimate of $G_i[k]x[k]$.
The expression \eqref{def_cov_epsilon} enables us to consider
the matrix sequence $\bm{W}[k]$ that satisfies the following recursive equation 
\begin{equation}\label{eq:defW}
	\bm{W}[k+1]=\bm{\Pi}[k] \bm{W}[k] \bm{\Pi}^{\top}[k] +\bm{Q}[k] .
\end{equation}
The following lemma shows that $\bm{W}[k]$ is non-singular.
\begin{lemma}\label{lm:cond_num}
	If Assumption \ref{as:sample_time} is satisfied and $\bm{W}[k]$ is initialized to be non-singular, i.e., $\bm{W}[0] \succ 0$, then there exists a positive constant $\overline{W}$ such that for all $k\in\Zb_{\geq 0}$,
	\begin{equation}
		0 \prec \bm{W}[k] \preceq \overline{W}\cdot I ,
	\end{equation}
	where $I$ is the identity matrix of size $mn\times mn$. \QET
\end{lemma}

\begin{proof}
    See Appendix~\ref{lmpf:cond_num}.
\end{proof}

Lemmas~\ref{lm:span} and \ref{lm:cond_num} show the non-singularity of matrices $V_i[k]$ and $\bm{W}[k]$, enabling us to
propose a state estimation fusion 
that provides a state estimate $x_{\ls}[k]$ by solving the following least squares problem:
\begin{subequations}\label{pb:least_square}
	\begin{align}
		\underset{{x_{\ls}}[k], \theta[k]}{\text{minimize}}&\quad \frac{1}{2} \theta[k]^{\top} \bm{\tilde W}^{-1}[k] \theta[k]  \label{pb:least_square_obj} \\
		\text { subject to}&\quad
		\bm{V}[k] \bm{\zeta} [k]=
		\bm{H} x_{\ls}[k] + \theta[k]  \label{pb:least_square_cons}
	\end{align}
\end{subequations}
where $\bm{\zeta}[k] \triangleq \big[ \zeta_1^\top[k],\,\zeta_2^\top[k],\dots,\,\zeta_m^\top[k] \big]^\top, ~ \bm{H} \triangleq \big[ H_1^\top,\,H_2^\top,\dots,\,H_m^\top \big]^\top, ~ \tilde{\bm{W}}[k] \triangleq \bm{V}[k] \bm{W}[k] \bm{V}^\top[k],~ \bm{V}[k] \triangleq \text{blkdiag}(V_1[k],\,V_2[k],\dots,\,V_m[k])$, and blkdiag$(\cdot)$ stands for a block diagonal matrix. Note that $\zeta_i[k]$ is defined in \eqref{eq:def_zeta} while matrices $V_i[k]$ and $H_i$ are defined in Lemma \ref{lm:span}.
The following theorem shows that the minimizer $x_{\ls}[k]$ of \eqref{pb:least_square} can exactly recover the Kalman state estimate $\hat{x}[k]$ based on local estimators $\zeta_{i}[k]$. 

\begin{theorem}\label{th:least_square}
    Suppose that $x_{\ls}[k]$ is the solution to \eqref{pb:least_square} and $\bm{W}[0]$ is
    a strictly positive definite Hermitian matrix. Then, the solution $x_{\ls}[k]$ is identical to the asynchronous sampled-data Kalman state estimate $\hat{x}[k]$ defined in \eqref{eq:def_xt_asy}, i.e., $x_{\ls}[k]=\hat{x}[k]$.
    \QET
\end{theorem}
\begin{proof}
    See Appendix~\ref{ap:th_least_square}.
\end{proof}

\begin{remark}
    Instead of directly computing the state estimate in \eqref{eq:sum_zeta=xhat}, we solve the least squares problem \eqref{pb:least_square} to obtain the state estimate, which yields the same result, as stated in Theorem~\ref{th:least_square}. Although the least squares \eqref{pb:least_square} is more complex, it has a decentralized form and its improved version, introduced in the next section, is secure against spatio-temporal attacks.
    \QET
\end{remark}

\section{Secure State Estimator}
\label{sec:secure_algorithm}
In this section, we propose a secure state estimation algorithm against $p$-sparse spatio-temporal false data attacks introduced in Definition~\ref{def:attack}. 
Prior to the algorithm, we present an analysis of spatio-temporal attacks.
\subsection{Attack analysis}
\label{sec:attack_analysis}
In this subsection, we carry out an analysis of spatio-temporal attacks to show how malicious activities impact state estimates.

\textbf{False-data injection}: this attack {retains} correct time-stamps, but manipulates measurements. More specifically, at sampling-time $k$, one can formulate the false-data injection as follows:
\begin{align}
    y_i^a[k] \triangleq y_i[k] + a_i[k], ~\text{if}~ \phi_i[k] = 1 ~\text{and}~i \in \Cc,
\end{align}
where $y_i^a[k]$ is the attacked measurement, $y_i[k]$ is the correct measurement, and $a_i[k]$ is the attack signal. The correct time-stamp guarantees the correctness of $\Pi[k-1]$ in \eqref{eq:def_zeta}. As a consequence, the impact of the false data injected into the measurement can be described in the local estimator as follows:
\begin{align}
    \zeta_i[k] \triangleq \zeta_i^o[k] + \zeta_i^f[k], \label{attack_local_est}
\end{align}
where $\zeta_i^o[k]$ is the oracle local estimator computed by \eqref{eq:def_zeta} and $\zeta_i^f[k] \triangleq \Ib_{\Zb_{>0}}(k) \sum_{\ell=0}^{k-1} \big( \prod_{p = 0}^{k-1-\ell} \Pi[k-1-p] \big) K_i[\ell]a_i[\ell] + K_i[k] a_i[k]$ is the malicious impact. Denote $\bm{\zeta}^o[k] \triangleq \big[ \zeta_1^{o\top}[k],\ldots,\zeta_m^{o\top}[k] \big]^\top$.

\textbf{Time-stamp manipulation}: this attack {retains} the correct measurement, but manipulates the time-stamp from the correct time-stamp $t$ to the attacked time-stamp $t^a~(t^a \neq t)$. Although the measurement $y_i(t)$ remains unchanged, the time-stamp manipulation consequently forces the estimator to treat $y_i(t)$ at the attacked time-stamp $t^a$. As a consequence, there is a mismatch of the measurement at time-stamp $t^a$, which is $y_i(t) - y_i(t^a)$. One can formulate the measurement at time $t^a$ received by the estimator as follows: $y_i(t) = y_i(t^a) + \big( y_i(t) - y_i(t^a) \big), ~\text{if}~i \in \Cc$.
{Although all the discretized matrices may change due to the manipulated time-stamp, this change is common to all local estimators, whereas the malicious effect of time-stamp manipulation is converted into the measurement mismatch, which can be modeled as a false-data injection attack.} Therefore, the malicious impact can be described as \eqref{attack_local_est}.


\textbf{Denial-of-service}: this attack strategy, motivated by jamming attacks such as \cite{li2015jamming}, can be viewed as the time-stamp manipulation where the attacked time-stamp $t^a$ is set at infinity.

\textbf{False-data generation}: this attack strategy can be described as the combination of false-data injection and time-stamp manipulation. 
As a result, the malicious impact of the false-data generation can also be described as \eqref{attack_local_est}.

In summary, the malicious impact of spatio-temporal attacks can be formulated as the false data injected into the local estimators of the corrupted sensors in \eqref{attack_local_est}. This formulation enables us to design the secure fusion in the following.

\vspace{-0.0cm}
\subsection{Secure fusion}
In light of the previous analysis, the malicious impact of the attacks can be isolated at separate local estimators that correspond to corrupted sensors. This observation enables us to improve the least squares problem \eqref{pb:least_square} in the following secure fusion where its minimizer $\check x[k]$ is a secure state estimate:
\begin{subequations}\label{pb:least_square_secure}
	\begin{align}
		\underset{{\check{x}}[k], \, \mu[k], \, \vartheta[k]}{\text{minimize}}&\quad \frac{1}{2} \mu[k]^{\top} \, \tilde{\bm{W}}^{-1}[k] \, \mu[k]  + \gamma \norm{\vartheta[k]}_1
		\\
		\text { subject to}&\quad
		\bm{V}[k]\bm{\zeta}[k]=
		\bm{H} \check{x}[k]+\mu[k] + \vartheta[k].  
	\end{align}
\end{subequations}

In the remainder of this section, 
we analyze the minimizer $\check x[k]$ without and with spatio-temporal attacks. The analysis will take the solution to \eqref{pb:least_square} in the absence of attacks as ground truth, i.e., the solution $(x_\ls[k],\,\theta[k])$ obtained by solving \eqref{pb:least_square} with $\bm{\zeta}[k] = \bm{\zeta}^o[k]$.

Recall that for the least squares problem \eqref{pb:least_square}, its minimizer $\theta[k]$ (see Appendix~\ref{ap:th_least_square}) can be computed as: $\theta[k] = \big[ I - \bm{G}[k] \big( \textbf{1}_m^\top \otimes I \big) \big] \bm{\zeta}[k]$,
which enables us to evaluate the solution to the problem \eqref{pb:least_square_secure} in the absence of the attacks in the next theorem.
\begin{theorem} \label{th:LS_noattack}
	Consider the least squares problems \eqref{pb:least_square} and \eqref{pb:least_square_secure} with a given $\gamma > 0$, let $(x_{\ls}[k]$, $\theta[k])$ be the minimizer for the problem \eqref{pb:least_square} and $(\check{x}[k]$, $\mu[k]$, $\vartheta[k])$ be the minimizer for the problem \eqref{pb:least_square_secure}. 
	In the absence of attacks, if the following condition holds
	\begin{align}
        \gamma >
		\|  \tilde{\bm{W}}^{-1}[k] \theta[k] \|_{\ift},
		\label{th:condition_wt_noattack}
	\end{align}
	then 
	${\check{x}}[k]=x_{\ls}[k]$, $\mu[k] = \theta[k]$, and $\vartheta[k] = 0$. \QET
\end{theorem}
\begin{proof}
    See Appendix~\ref{thpf:LS_noattack}.
\end{proof}
Let us make use of the following definition of a function that will help us in evaluating the minimizer $\check x[k]$ of \eqref{pb:least_square_secure} against spatio-temporal attacks in the subsequent theorem.
\begin{definition}
    \label{def:h_a}
    Given an $n$-dimensional vector $x \in \Rb^n$ and a positive integer $a$, we define a function $h_a: \Rb^n \rightarrow \Rb$ such that $h_a(x)$ takes the $a$-th largest value of the vector $x$.
\end{definition}
\begin{theorem}[Secure fusion] \label{th:secure_fusion}
    Consider the least squares problems \eqref{pb:least_square} and \eqref{pb:least_square_secure} with a given $\gamma > 0$, let $(x_{\ls}[k], \, \theta[k])$ be the minimizer for the problem \eqref{pb:least_square} in the absence of attacks and $(\check x[k], \, \mu[k], \, \vartheta[k])$ be the minimizer for the problem \eqref{pb:least_square_secure} in the presence of attacks. 
	In the presence of attacks, the difference between $x_{\ls}[k]$ and $\check x[k]$ has the following upper bound for all $j \in J$:
    \begin{align}
     \big \lvert \big[ \check x[k] \big]_j \!-\! \big[ x_{\ls}[k] \big]_j \big  \rvert \!\leq \!  \max 
    \big \{ \big \lvert h_c \big( \eta^j[k] \big) \big \rvert , \big \lvert \!-\! h_c \big( -\eta^j[k] \big) \big \rvert \big \}, 
    \label{th:estimation_upperbound}
    \end{align}
    where the function $h_c(\cdot)$ is defined in Definition~\ref{def:h_a}, $\eta^j[k]$ is a $\lvert \Ec_j \setminus \Cc \rvert$-dimensional vector with its $i$-th element $[\eta^j[k]]_i \triangleq [\theta_i[k]]_j + \gamma \, \bm{e}_{n(i-1)+j}^\top \tilde{\bm{W}}[k]  \check \vartheta[k]~(\forall \, i \in \Ec_j \setminus \Cc)$, $\check \vartheta[k] \in \partial  \lVert \bm{V}[k] \bm{\zeta}^f[k] - \vartheta[k] \big \rVert_1$, and $c \triangleq \big \lceil \frac{\lvert \Ec_j \setminus \Cc \rvert - \lvert \Ec_j \bigcap \Cc \rvert}{2} \big \rceil$. \QET
\end{theorem}
\begin{proof}
    See Appendix~\ref{thpf:secure_fusion}.
\end{proof}

It is worth noting that the vectors $\eta^j[k]$ in Theorem~\ref{th:secure_fusion} are independent of the information provided by the attacked sensors for all $j \in \Jc$ by definition and $\check \vartheta[k]$ is bounded in the range $[-1,1]$. Consequently,
the result of Theorem~\ref{th:secure_fusion} shows us that the upper bound of the estimation error under spatio-temporal attacks, which is $\big \lvert \big[ \check x[k] \big]_j - x[k] \big]_j \big  \rvert \leq \big \lvert \big[ \check x[k] \big]_j - \big[ x_{\ls}[k] \big]_j \big  \rvert + \big \lvert \big[ x_{\ls}[k] \big]_j - \big[ x[k] \big]_j \big  \rvert$,  is independent of the malicious activities for all $j \in \Jc$. Thus, the secure state estimate $\check x[k]$ is resilient to such attacks, satisfying uniform error bound \eqref{uni_bound}.

{For implementation, Algorithm~\ref{alg:secure_fusion} describes a step-by-step procedure to construct the $\ell_1$-regularization least squares problem \eqref{pb:least_square_secure}, whose solution is the secure state estimate $\check x[k]$, resilient to spatio-temporal attacks according to Theorem~\ref{th:secure_fusion}.}
{
\begin{remark}
    In practice, it is not necessary to compute the right-hand side of \eqref{th:condition_wt_noattack} online. An effective approach is to choose a baseline $\gamma_0$ and its scaled values $\{\gamma\}$ with $\gamma = \rho \gamma_0$, where $\rho > 1$, e.g., some $\rho \in [\underline{\rho},~\bar \rho]$ with $\bar \rho > \underline{\rho} > 1$. 
    The chosen set of regularization parameters $\{\gamma\}$ gives a set of decoupled optimization problems in the form of \eqref{pb:least_square_secure}, which we solve in parallel.
   If the results from solving \eqref{pb:least_square_secure} with two different values of $\gamma$ are the same, these two values of $\gamma$ lead to the same optimizer $\vartheta^\star[k] = 0$. Therefore, we can choose either one of these two values of $\gamma$.  \QET
\end{remark}
\begin{remark}
    The median-based fusion used in \cite{li2023secure} ensures bounded covariance of the secure estimate, but does not establish its relationship to the oracle Kalman estimate, which is optimal by definition. In contrast, the present fusion \eqref{pb:least_square_secure} has two distinctive properties: (i)  it exactly recovers the oracle Kalman estimate in the absence of attacks (Theorem~\ref{th:LS_noattack}); and (ii) under attacks, it provides a component-wise, attack-independent error bound tied to the oracle Kalman performance (Theorem~\ref{th:secure_fusion}). \QET
\end{remark}
}

\begin{algorithm}[t]
\color{blue}
\captionsetup{font={color=blue}}
\caption{Secure fusion under spatio-temporal attacks}
\label{alg:secure_fusion}
\begin{algorithmic}[1]
  \Statex{\textbf{Input:}} system matrices $A, C, Q$, $R$, and $H_i$; the measurement availability index $\phi_i[k]$; and the $\ell_1$-regularization parameter $\gamma$ 
  \Statex{\textbf{Output:}} secure state estimate $\check x[k]$
  \For {each discrete-time index $[k]$} 
  \State {compute discretized matrices in \eqref{sampled_matrices}, \eqref{eq:def_Kk_asy}, and \eqref{eq:def_Pi}}
  \For {each sensor index $i$}
  \State{
  \multiline{compute $\zeta_i[k]$ in \eqref{eq:def_zeta}, $G_i[k]$ in \eqref{eq:defG}, $V_i[k]$ in Lemma~\ref{lm:span}.}
  }
  \EndFor
  \State{compute $\tilde{\bm{W}}[k] = \bm{V}[k] \bm{W}[k] \bm{V}^\top[k]$, where $\bm{W}[k]$ in \eqref{eq:defW}}
  \State{solve \eqref{pb:least_square_secure}}
  \State \Return{$\check{x}[k]$}
  \EndFor
\end{algorithmic}
\end{algorithm}

\vspace{-5pt}
\section{Simulation Results}
To validate the obtained results, the proposed secure state estimation algorithm\footnote{Code is available at \href{https://github.com/tungnguyenkstnk58/secure-asynchronous-state-estimation}{https://tinyurl.com/sec-asyn-est}} \eqref{pb:least_square_secure} is implemented in the IEEE 14-bus system \cite{korres2010robust}, which contains 28 state variables (a phase angle and an angular frequency for each bus) and 42 sensors (electric power, phase angle, and angular frequency sensors for each bus). 
By {leveraging} the structure of the block diagonal $\bm{V}[k]$ and the asynchronous sampling, the elements of $\bm{V}[k] \bm{\zeta}[k]$ in \eqref{pb:least_square_secure}, which correspond to off-sampling sensors, are true since they are computed based on system modeling. Consequently, we can set elements of $\vartheta[k]$ in \eqref{pb:least_square_secure}, which correspond to those true values, are zero in the implementation.
Theorems~\ref{th:least_square}--\ref{th:secure_fusion} are validated in the following.
In the first scenario, we conduct the state estimation using the KF \eqref{eq:asy_kalman}, the proposed least squares \eqref{pb:least_square}, and the proposed secure least squares \eqref{pb:least_square_secure} with two different values of $\gamma$, i.e., $\gamma = 2$ and $\gamma = 400$, in the absence of attacks {(see Fig.~\ref{fig:14busrms_noattack} for estimation errors). The following three curves overlap: state estimates provided by using the KF \eqref{eq:asy_kalman} (blue), solving \eqref{pb:least_square} (red), and solving \eqref{pb:least_square_secure} (purple) with $\gamma$ satisfying \eqref{th:condition_wt_noattack}, which is consistent with Theorems \ref{th:least_square}--   \ref{th:LS_noattack}. In contrast, for a small $\gamma$, the sufficient condition in Theorem 2 is not satisfied and a deviation is expected, as illustrated by the yellow dashed line in Fig.~\ref{fig:14busrms_noattack}.} 
\begin{figure}[!t]
    \centering
    \includegraphics[width=0.78\linewidth]{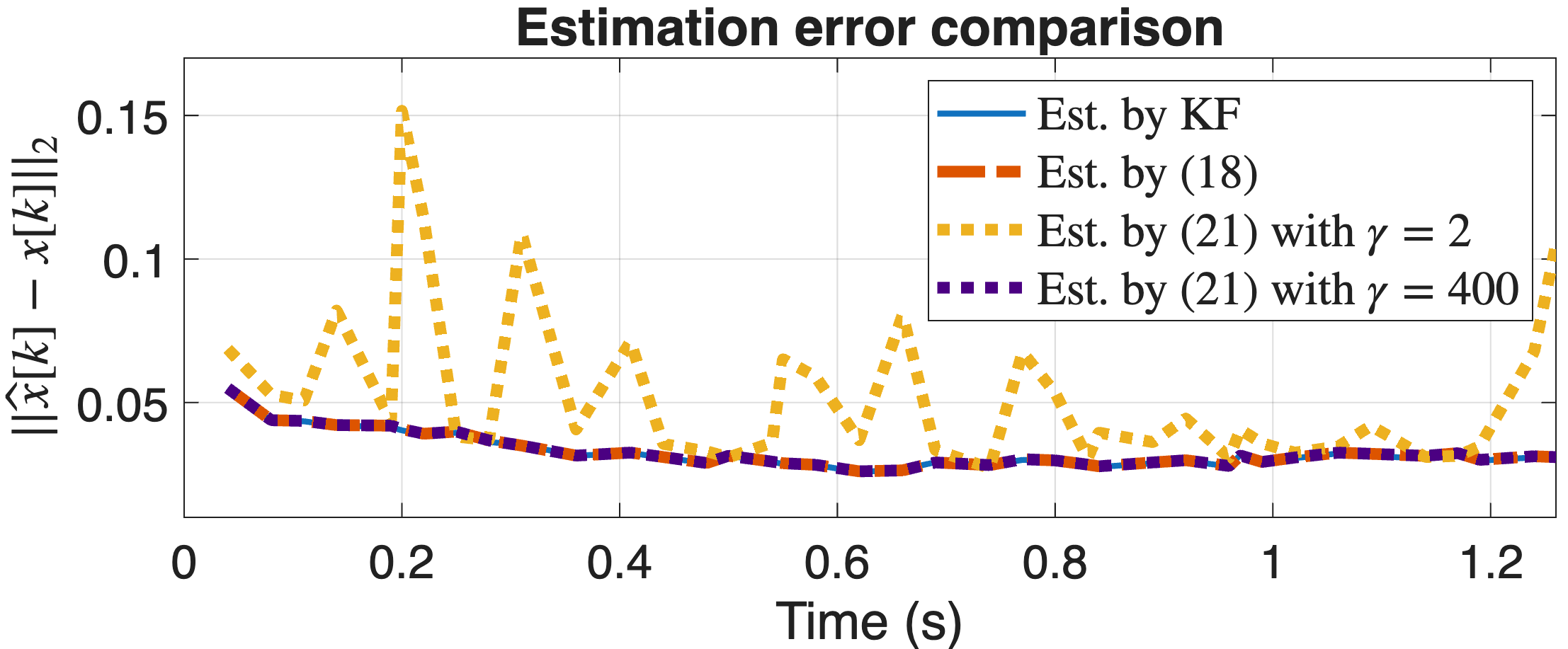}
    \caption{Estimation error comparison using the sampled-data KF \eqref{eq:asy_kalman}, the least squares problem \eqref{pb:least_square}, and the least squares problem \eqref{pb:least_square_secure} with two different values of $\gamma$ in the absence of attacks. 
    }
    \vspace{-8pt}
    \label{fig:14busrms_noattack}
\end{figure}

It remains to validate the result of Theorem~\ref{th:secure_fusion}.
In the second scenario, we conduct spatio-temporal attacks: false data injection on the phase angle sensor of bus $3$, time-stamp manipulation on the power sensor of bus $5$, DoS on the angular frequency sensor of bus $4$, and false data generation on the power sensor of bus $2$ (see Fig.~\ref{fig:14busatt_activ}). {To avoid bias toward any specific adversarial strategy, all attacks are instantiated with randomly generated false-data values $(a, y^f \sim \Nc(2,1))$, randomly manipulated time-stamps $\big( (t^a - t) \sim \Uc[0.01, 0.05] \big)$, and complete packet-drop events.} The state estimates provided by the sampled-data KF \eqref{eq:asy_kalman} without attacks and the secure least squares problem \eqref{pb:least_square_secure} with $\gamma =2$ under attacks are illustrated in Fig.~\ref{fig:14busstate_4attack}. 
{For further performance comparison, we also implement a robust KF \cite{li2023improved} under the same attack scenario.}
A clearly observable difference between the estimation errors of all these estimation algorithms is witnessed in Fig.~\ref{fig:14busrms_4attack}. While the state estimate provided by the secure least squares problem \eqref{pb:least_square_secure} is resilient to the attacks, that provided by the sampled-data KF exhibits a very large error. {Further, \eqref{pb:least_square_secure} also provides a better estimation performance than the robust KF \cite{li2023improved}}. These results illustrate the effectiveness of our proposed secure state estimation algorithm.
\begin{figure}[!t]
    \centering
    \includegraphics[width=0.75\linewidth]{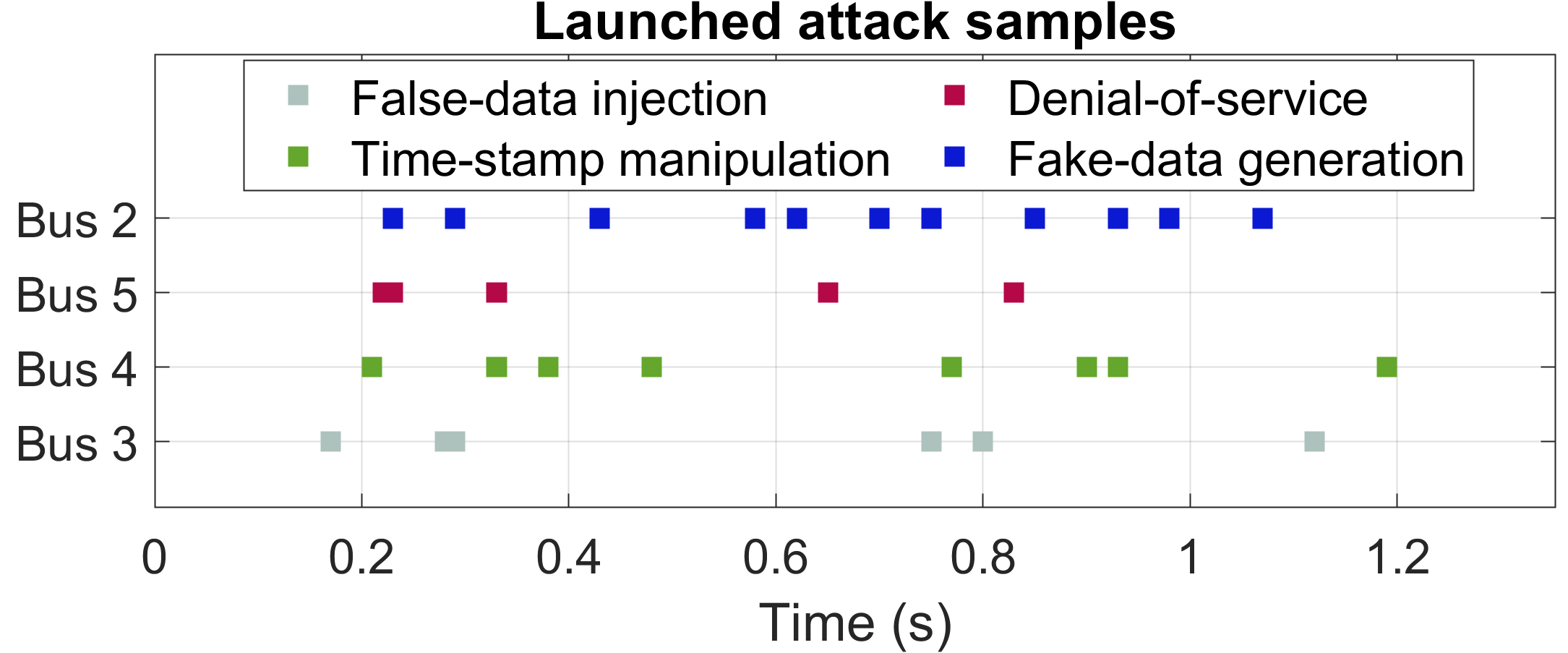}
    \caption{The spatio-temporal attacks are launched on sensors of buses 2--5 with false data injection on the phase angle sensor of bus 3, time-stamp manipulation on the power sensor of bus 5, denial-of-service on the angular frequency sensor of bus 4, and completely new false data generation on the power sensor of bus 2.}
    \vspace{-8pt}
    \label{fig:14busatt_activ}
\end{figure}
\begin{figure}[!t]
    \centering
    \includegraphics[width=1.0\linewidth]{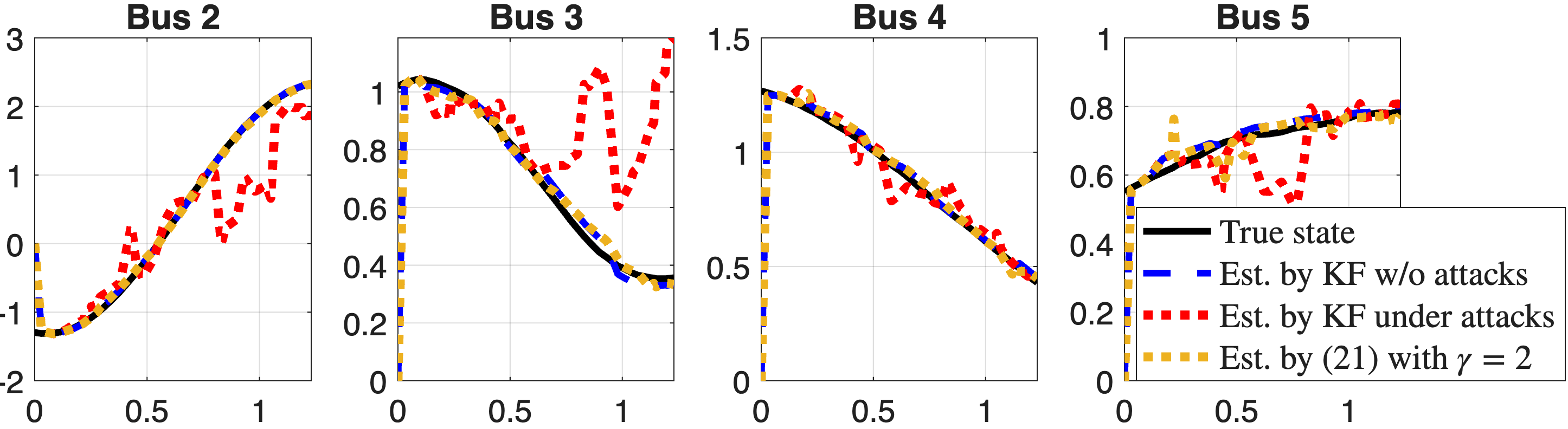}
    \caption{
    The horizontal axes represent time in seconds. The least-squares problem \eqref{pb:least_square_secure} provides a resilient state estimate against the attacks while the KF fails to provide a resilient state estimate.}
    \vspace{-10pt}
    \label{fig:14busstate_4attack}
\end{figure}
\begin{figure}[!t]    
    \centering
    \includegraphics[width=0.81\linewidth]{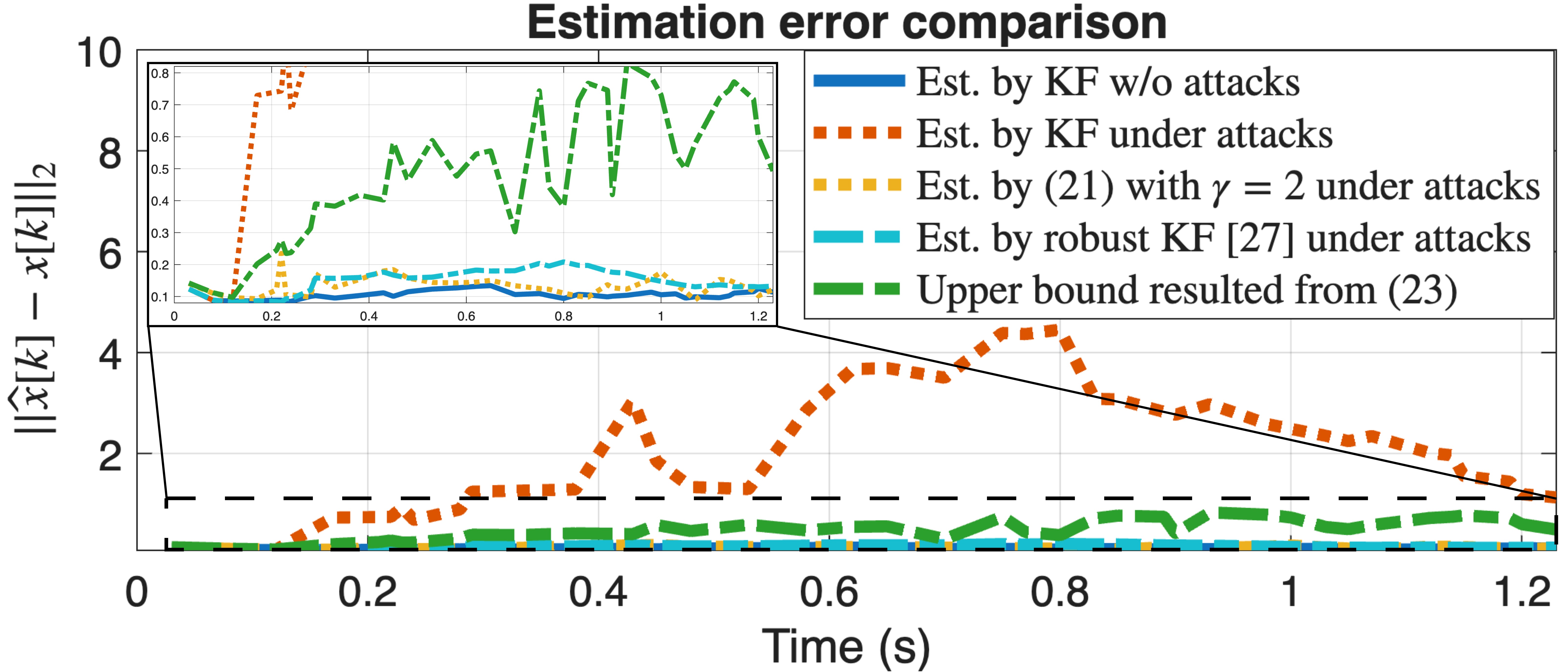}
    \caption{Estimation error comparison between the KF \eqref{eq:asy_kalman}, the secure least squares problem \eqref{pb:least_square_secure}{, and the robust KF \cite{li2023improved}. 
    The inset shows a zoomed time interval with a different $y$-axis scale, highlighting that the estimate provided by solving \eqref{pb:least_square_secure} with $\gamma = 2$ remains close to the oracle KF estimate without attacks and stays below the theoretical upper bound, validating Theorem~\ref{th:secure_fusion}. 
     }
    }
    \label{fig:14busrms_4attack}
    \vspace{-0.5cm}
\end{figure}
\section{Conclusion}
This paper presented a secure state estimation algorithm for continuous LTI systems with non-periodic and asynchronous measurements under spatio-temporal false data attacks. The secure state estimation was developed based on the decomposition of the sampled-data KF to provide a state estimate that is exactly the same as the one provided by the sampled-data KF in the absence of attacks and resilient to spatio-temporal false data attacks. The effectiveness of the proposed secure state estimation was validated through an IEEE benchmark for power systems. 
{Combining a robust Kalman filtering method with the proposed $\ell_1$-regularization fusion to jointly address attacks and measurement outliers is a promising direction for future work.}

\renewcommand{\thesection}{A.\arabic{section}}
\setcounter{section}{0}  
%
%
\vspace{-0.3cm}
\section*{Appendix}
\section{Proof of Lemma~\ref{lm:cond_num}}
\label{lmpf:cond_num}
Before showing the proof of Lemma~\ref{lm:cond_num}, we present the following supporting results. Their proofs can be found in \cite{li2025secure}.
\begin{lemma}\label{lm:sumW}
Denote sub-blocks $\bm{W}_{i j}[k] \in \mathbb{Rb}^{n \times n}$ where $\bm{W}[k]=\left(\bm{W}_{i j}[k]\right)_{m \times m}$. For all $i\in\Ic$ and arbitrary time $k\in\Zb_{\geq 0}$, we have $\sum_{j=1}^{m}  \bm{W}_{i j}[k]=P[k] G_i^{\top}[k]$. \QET
\end{lemma}

\begin{proposition}[Stability of Asynchronous KF~\cite{Feddaoui_asy_kalman}]\label{pp:central_kalman}
    Suppose that Assumptions \ref{as:R} and \ref{as:sample_time} hold.
    The estimation covariance $P(t)$ of the sampled-data KF defined in \eqref{eq:asy_kalman} satisfies the following properties:
	\begin{equation}
		\underline{p} I \preceq P(t)\preceq\overline{p} I, \, \forall \, t\geq0 ,
	\end{equation}
	where $\underline{p},\,\overline{p}$ are constant scalars regardless of the sampling times. \QET
\end{proposition}

\textit{Proof of Lemma~\ref{lm:cond_num}}:
	We first prove the upper bound. Based on Lemma \ref{lm:sumW}, we have $\sum_{j=1}^{m}\bm{W}_{ij}[k]=P[k]G^{\top}_i[k]$ for all $k\in\Zb_{\geq 0}$. Summing both sides over $i$ and recalling that $\sum_{i=1}^{m}G_i[k]=I$ from Lemma \ref{lm:sumG}, one obtains $\sum_{i=1}^{m}\sum_{j=1}^{m}\bm{W}_{ij}[k]=P[k]$,
	where $P[k]$ is the estimation covariance of asynchronous Kalman estimator defined in \eqref{eq:def_Pt_asy}.
    On the other hand,
    the result of Proposition \ref{pp:central_kalman} can yield $\underline{\alpha}\cdot I \preceq P[k] \preceq \overline{\alpha}\cdot I$,
    resulting in that
    for each index $i$, the diagonal block satisfies
	$\bm{W}_{ii}[k]\preceq \overline{\alpha} I$ considering that every block $\bm{W}_{ij}[k]$ is positive semi-definite.
	As a result, there exists a constant $\overline{W}$ such that $\bm{W}[k]\preceq \overline{W}\cdot I$ holds for all time index $k$. \hfill \QED   
\vspace{-0.3cm}
\section{Proof of Theorem \ref{th:least_square}} \label{ap:th_least_square} 


    Firstly, since $\bm{V}[k]$ is non-singular based on Lemma~\ref{lm:span}, we multiply both sides of \eqref{pb:least_square_cons} with $\bm{V}^{-1}[k]$. After using the notations $\bm{G}[k] \triangleq \big[ G_1^\top[k],\,G_2^\top[k],\ldots,\,G_m^\top[k] \big]^\top$ and $\hat \theta[k] = \bm{V}^{-1}[k] \theta[k]$, one obtains the following least squares problem, equivalent to \eqref{pb:least_square}:
    \begin{subequations}\label{pb:least_square_mod}
	\begin{align}
		\underset{{x_{\ls}}[k], \hat \theta[k]}{\text{minimize}}&\quad \frac{1}{2} \hat \theta[k]^{\top} \bm{W}^{-1}[k] \hat \theta[k]   \\
		\text { subject to}&\quad
		\bm{\zeta} [k]=
		\bm{G}[k] x_{\ls}[k]+ \hat \theta[k] .  
	\end{align}
    \end{subequations}

	Secondly, we prove that if the initial value of $\bm{W}[k]$ is Hermitian and strictly positive definite and satisfies $\sum_{j=1}^{m}  \bm{W}_{ij}[0]=\Sigma\cdot G_i^{\top}[0]$,
	for all $i\in\Ic$, then Theorem \ref{th:least_square} holds.
	This initialization $\bm{W}[0]$
 and Lemma \ref{lm:sumW} imply that the following holds for all $k\in\Zb_{\geq 0}$:
	\begin{equation}\label{eq:FW=PG}
		\begin{bmatrix}
			I & \cdots & I
		\end{bmatrix} \bm{W}[k]=P[k]
        \big[ G_{1}^{\top}[k], \cdots, G_{m}^{\top}[k] \big],
	\end{equation}
    resulting in 
    $
	\bm{G}^{\top}[k]\bm{W}^{-1}[k]\bm{G}[k]=P^{-1}[k]\begin{bmatrix}
		I & \cdots & I
	\end{bmatrix}\bm{G}^{\top}[k].
	$
    
	On the other hand, the solution to \eqref{pb:least_square_mod} is given by 
	\begin{align*}
		x_{\ls}[k]=\big( \bm{G}^{\top}[k]\bm{W}^{-1}[k]\bm{G}[k]\big)^{-1}\bm{G}^{\top}[k]\bm{W}^{-1}[k]\bm{\zeta}[k].
	\end{align*}
	Since $\sum_{i=1}^{m}G_i[k]=I$ from Lemma~\ref{lm:sumG}, we finally conclude that
	\begin{align*}
		x_{\ls}[k]=P[k]\bm{G}^{\top}[k] \bm{W}^{-1}[k]\bm{\zeta}[k]=
		\begin{bmatrix}
			I& \cdots&  I
		\end{bmatrix}\bm{\zeta}[k]=\hat{x}[k],
	\end{align*}
	where the second and third equalities come from \eqref{eq:FW=PG} and \eqref{eq:sum_zeta=xhat}, respectively.
    Finally, the existence of $\bm{W}[0]$ is shown in \cite[A.9]{li2025secure}. \hfill \QED
\vspace{-0.3cm}
\section{Proof of Theorem~\ref{th:LS_noattack}}
\label{thpf:LS_noattack}
Let us introduce two new variables $\alpha[k]$ and $\beta[k]$ as the deviations between the two solutions to the problems \eqref{pb:least_square} and \eqref{pb:least_square_secure} such that $\alpha[k] \triangleq \check{x}[k] - x_{\ls}[k]$ and $\beta[k] \triangleq \mu[k] -  \theta[k]$. 
    The proof will be completed if we show that $\alpha[k] = 0$ and $\beta[k] = 0$ in the absence of the attacks. It is worth noting that the absence of the attacks implies
    the same $\bm{\zeta}[k]$ in \eqref{pb:least_square} and \eqref{pb:least_square_secure}, resulting in $\bm{H} \alpha[k] + \beta[k] + \vartheta[k] = 0$.
    As a result of utilizing the new deviation variables $\alpha[k]$ and $\beta[k]$, solving \eqref{pb:least_square_secure} is equivalent to solving the following problem:
    \begin{align}
        \underset{{\alpha}[k], \beta[k], \vartheta[k]}{\text{minimize}}&
        \frac{1}{2} \beta[k]^\top \tilde{\bm{W}}^{-1} [k]\beta[k] 
        + \theta[k]^\top \tilde{\bm{W}}^{-1}[k]  \beta[k] + \gamma \norm{\vartheta[k]}_1 \notag \\ 
        \text{subject to}~&~ \bm{H} \alpha[k] + \beta[k] + \vartheta[k] = 0.  \label{pb:lasso_reform}
    \end{align}
    Let us consider the second term of the objective function \eqref{pb:lasso_reform} which can be rewritten based on its constraint as follows:
    \begin{align}
        \theta[k]^\top \tilde{\bm{W}}^{-1}[k]   \beta[k] 
        = - \theta[k]^\top \tilde{\bm{W}}^{-1}[k] \vartheta[k], \label{pr:deviation}
    \end{align}
    where the equality comes from the fact that 
    $\theta[k]^\top \tilde{\bm{W}}^{-1}[k] \bm{H} = 0$ based on the KKT condition of the problem \eqref{pb:least_square}.
    On the other hand, for an arbitrary $\vartheta[k]$, \eqref{th:condition_wt_noattack} implies
    \begin{align}
        \gamma \norm{\vartheta[k]}_1 \geq 
        \theta[k]^\top \tilde{\bm{W}}^{-1}[k]  \vartheta[k],
    \end{align}
    where the equality occurs if, and only if, $\vartheta[k] = 0$. This result together with \eqref{pr:deviation} implies that the minimum value of the objective \eqref{pb:lasso_reform} is zero if, and only if, $\tilde{\bm{W}}^{-1} [k] \beta[k] =0$ and $\vartheta[k] = 0$. The empty null space of $\tilde{\bm{W}}^{-1} [k]$ gives us $\beta[k] = 0$.
    As a consequence, the constraint \eqref{pb:lasso_reform} results in $\alpha[k] = 0$ since the matrix $\bm{H}$ has an empty null space. The proof is completed. \hfill \QED
\vspace{-0.3cm}
\section{Proof of Theorem~\ref{th:secure_fusion}}
\label{thpf:secure_fusion}
Recall the analysis in Section~\ref{sec:attack_analysis} and \eqref{attack_local_est}, let us consider the problem \eqref{pb:least_square} with the oracle local estimator $\bm{\zeta}^o[k]$. 
Note that the fusion does not know the oracle value of $\bm{\zeta}^o[k]$ to find the optimal state estimate $x_{\ls}[k]$. However, this optimal solution $(x_{\ls}[k],\, \theta[k])$ can be utilized as a ground truth. On the other hand, we consider the problem \eqref{pb:least_square_secure} in the presence of the attacks, i.e., the corrupted local estimator $\bm{\zeta}[k] \neq \bm{\zeta}^o[k]$ where $\zeta_i[k]$ is defined in \eqref{attack_local_est}. 

Let us reuse the two deviation variables $\alpha[k]$ and $\beta[k]$ in Appendix~\ref{thpf:LS_noattack} such that
$\alpha[k] \triangleq \check{x}[k] - x_{\ls}[k]$ and $\beta[k] \triangleq \mu[k] -  \theta[k]$. 
Next, we plan to show that $\norm{\alpha[k]}$ lies in a small ball centered at the origin and is independent of the malicious activities. The constraint of \eqref{pb:least_square_secure} in the presence of attacks and the constraint of \eqref{pb:least_square} in the absence of attacks give us the following relationship: $\beta[k] = \bm{V}[k] \bm{\zeta}^f[k] - \bm{H} \alpha[k] -  \vartheta[k]$.
As a consequence, solving \eqref{pb:least_square_secure} in the presence of attacks is equivalent to solving the following problem:
\begin{align}
        \underset{{\alpha}[k], \beta[k],  \vartheta[k]}{\text{minimize}}&
        \frac{1}{2} \beta[k]^\top \tilde{\bm{W}}^{-1} [k]\beta[k] + \theta[k]^\top \tilde{\bm{W}}^{-1}[k]  \beta[k]
         + \gamma \norm{  \vartheta[k]}_1  \notag \\ 
        \text{subject to}~&~ \beta[k] = \bm{V}[k] \bm{\zeta}^f[k]  
        - \bm{H} \alpha[k] -  \vartheta[k]. 
        \label{pb:ls_wattack}
\end{align}
Let us denote $\hat \vartheta[k] \triangleq \bm{V}[k] \bm{\zeta}^f[k] -   \vartheta[k]$ and $\check \vartheta[k] \in \partial\big \lVert \hat \vartheta[k] \big \rVert_1$, i.e., the sub-gradient with respect to $\vartheta[k]$. It is worth noting that the $i$-th element of the sub-gradient $\check \vartheta[k]$ takes a value between $-1$ and $1$. Since $\bm{V}[k]$ is a block diagonal matrix, we have the following property: $\bm{V}_i[k] \bm{\zeta}^f_i[k] \neq 0$ when $i \in \Cc$ and $\bm{V}_i[k] \bm{\zeta}^f_i[k] = 0$ otherwise.

With the help of the KKT condition of \eqref{pb:least_square}, which is $\theta[k]^\top \tilde{\bm{W}}^{-1}[k] \bm{H} = 0$,
the optimization problem \eqref{pb:ls_wattack} can be solved by considering the following optimization problem:
\begin{align}
    \underset{\alpha[k],\, \hat \vartheta[k]}{\text{minimize}}~ L(\alpha[k],\,\hat \vartheta[k]),
    \label{pb:ls_wattack_obj}
\end{align}
where $L(\alpha[k],\hat \vartheta[k]) \triangleq \theta[k]^\top  \tilde{\bm{W}}^{-1}[k] \hat \vartheta[k] + \gamma \norm{\bm{V}[k] \bm{\zeta}^f[k]  -  \hat \vartheta[k]}_1$
\vspace{-10pt}
\begin{align}
    & \hspace{40pt} + \frac{1}{2} \big( \hat \vartheta[k] - \bm{H} \alpha[k] \big)^\top  \tilde{\bm{W}}^{-1}[k] \big( \hat \vartheta[k]  - \bm{H} \alpha[k] \big). \label{Lfunction}
\end{align}
Let us denote $(\hat \vartheta^\star[k],\,\alpha^\star[k])$ as the solution to \eqref{pb:ls_wattack_obj}, which satisfies
\begin{align}
    &\frac{\partial L(\alpha[k],\,\hat \vartheta[k])}{\partial \hat \vartheta[k]} (\hat \vartheta^\star[k],\,\alpha^\star[k]) =   \hat \vartheta^\star[k]  - \bm{H} \alpha^\star[k]  \notag \\
    &\hspace{2.5cm}
    +  \theta[k] + \gamma \tilde{\bm{W}}[k] \check \vartheta[k] = 0. \label{gradient_Lv}
\end{align}
Then, substituting \eqref{gradient_Lv} into \eqref{Lfunction} and leveraging $\theta[k]^\top \tilde{\bm{W}}^{-1}[k] \bm{H}  = 0$, which is the KKT condition of \eqref{pb:least_square}, give us the following:
\begin{align}
    &L(\hat \vartheta^\star[k],\,\alpha^\star[k]) = 
     -\frac{1}{2}\theta^\top[k] \tilde{\bm{W}}^{-1}[k]  \theta[k] 
    + \frac{1}{2} \gamma^2 \check \vartheta^\top[k] \tilde{\bm{W}}[k] \check \vartheta[k] 
    \notag \\
    &\textstyle + \gamma \sum_{i \, \in \, \Cc} ~ \sum_{j \in \Jc} \big \lvert \big[ V_i[k]\zeta^f_i[k]  + \theta_i[k]  \big]_j - \big[ H_i \alpha^\star[k] \big]_j  \notag \\ 
    & \textstyle
    + \gamma \, \bm{e}_{n(i-1)+j}^\top  \tilde{\bm{W}}[k]  \check \vartheta[k]  \big \rvert 
    + \gamma \sum_{i \, \in \, \Ic \setminus \Cc} ~ \sum_{j \in \Jc} \big \lvert [\theta_i[k]]_j 
    \notag \\
    &
    + \gamma  \, \bm{e}_{n(i-1)+j}^\top \tilde{\bm{W}}[k] \check \vartheta[k] 
    - \big[ H_i \alpha^\star[k] \big]_j \big \rvert.
\end{align}
For the state with index $j$, let us recall the index set of sensors that observe state $j$, which was denoted as $\Ec_j$ in \eqref{eq:def_Oc}, and the structure of the matrix $H_i$ in Lemma~\ref{lm:span}, resulting in
\begin{align}
    \big[ H_i \alpha^\star[k] \big]_j = \begin{cases}
        [\alpha^\star[k]]_j, ~~ \text{if}~ i \, \in \, \Ec_j, \\
        0,~~ \text{otherwise},
    \end{cases}
\end{align}
where $[\alpha^\star[k]]_j$ is the $j$-th element of $\alpha^\star[k]$. In the following,
we consider the function $L_j(\hat \vartheta^\star[k],\,\alpha^\star[k])$ that is a collection of terms containing $[\alpha^\star[k]]_j$ in $L(\hat \vartheta^\star[k],\,\alpha^\star[k])$ as follows:
\begin{align}
    &\textstyle L_j(\hat \vartheta^\star[k],\,\alpha^\star[k]) = 
    \sum_{i \, \in \, \Ec_j \bigcap  \Cc} \, \big \lvert [ V_i[k] \zeta^f_i[k]  + \theta_i[k] ]_j -  [\alpha^\star[k]]_j \notag \\ 
    & \textstyle +  \gamma  \, \bm{e}_{n(i-1)+j}^\top \tilde{\bm{W}}[k] \check \vartheta[k]  \big \rvert  
    + \sum_{i \, \in \, \Ec_j \setminus \Cc} \, \big \lvert [\theta_i[k]]_j
    \notag \\ 
    & 
    + \gamma \, \bm{e}_{n(i-1)+j}^\top \tilde{\bm{W}}[k]  \check \vartheta[k]
    -  [\alpha^\star[k]]_j \big \rvert. \label{Lj}
\end{align}
Due to the fact that the system is $2p$-observable, one has $2 \lvert \Ec_j \bigcap \Cc \rvert \leq 2 \lvert \Cc \rvert \leq 2p < \lvert \Ec_j \rvert$, resulting in $\lvert \Ec_j \bigcap \Cc \rvert < \lvert \Ec_j \setminus \Cc \rvert$. This result implies that the number of $[\alpha^\star[k]]_j$ in the first term of \eqref{Lj} is less than that of $[\alpha^\star[k]]_j$ in the second term of \eqref{Lj}. This observation enables us to apply the result of 
\cite[Lemma 6]{li2025secure}
to \eqref{Lj} together with the definition $[\alpha^\star[k]]_j  = \big[ \check x[k] \big]_j - \big[ x_{\ls}[k] \big]_j$, resulting in \eqref{th:estimation_upperbound}. 
\hfill \QED

\bibliographystyle{IEEEtran}        
\bibliography{async_est_24.bib}

\newpage
\section{Proof of Lemma~\ref{lm:span}}
\label{pflm:span}
The second equality of \eqref{rowspan_Gi} holds by the definition of the observability matrix $O_i$ and $H_i$.
By induction method, we show the first equality of \eqref{rowspan_Gi}. Let us assume $\rs(G_i[k])= \rs(O_i)$. We need to show that $\rs(G_i[k+1])= \rs(O_i)$.

According to Assumption \ref{as:geo_mul}, one obtains 
$$\rs(G_i[k]A^{-1}[k])= \rs(O_iA^{-1}[k])=\rs(O_i)$$ and $\rs(\Pi[k]G_i[k]A^{-1}[k])= \rs(O_i)$. Moreover, since $\rs(K_i[k+1]C_i)\subseteq \rs(O_i)$ and one obtains that $\rs(G_i[k+1]) \subseteq \rs(O_i)$. 
    
If $((A[k]-K[k+1]CA[k])G_i[k] A^{-1}[k] +K_i[k+1]C_i )^\top \bm{e}_j = 0$, we alter $K_i[k+1]$ slightly so that the equation does not hold while the performance of the estimator is not influenced. As a result, $\rs((A[k]-K[k+1]CA[k])G_i[k] A^{-1}[k] +K_i[k+1]C_i)=\rs(G_i[k+1])$ and the proof is completed.
\hfill \QED

\section{Proof of Lemma~\ref{lm:sumG}}
\label{lmpf:sumG}
The proof is presented by the induction.
Assume that $\sum_{i=1}^{m} G_i[k]=I$ and we show that $\sum_{i=1}^{m} G_i[k+1]=I$:
\begin{align*}
    \sum_{i=1}^{m} G_i[k+1]
		=&\sum_{i=1}^{m} \Pi[k] G_i[k] A^{-1}[k]+K_i[k+1]C_i[k+1]\\
		=& \, \Pi[k]A^{-1}[k]+K[k+1]C = I,
\end{align*}
where the second equality comes from the assumption that $\sum_{i=1}^{m} G_i[k]=I$ and the last equality comes from the definition of $\Pi[k]$ in \eqref{eq:def_Pi}.
\hfill \QED

\section{Proof of Lemma~\ref{lm:epsilon}}
\label{lmpf:epsilon}
According to the definition of $\epsilon_i[k]$, we have
	\begin{align*}
		&\epsilon_{i}[k+1] =
		\zeta_{i}[k+1]-G_{i}[k+1] x[k+1] \\ 
		=\,& \Pi[k]\zeta_{i}[k] +K_i[k+1]\left(C_i A[k] x[k]+C_i w[k]+v_{i}[k+1]\right) \\
		&-G_{i}[k+1] \left(A[k] x[k]+ w[k]\right)\\
		=\,&\Pi[k] \zeta_{i}[k]- \left(G_{i}[k+1] A[k]-K_i[k+1] C_i A[k]\right) x[k] \\
		&-\left(G_{i}[k+1]-K_i[k+1] C_i\right) w[k]+K_i[k+1] v_{i}[k+1] \\
		=\,&\Pi[k] \left(\zeta_{i}[k]-G_{i}[k] x[k]\right)  -\Pi[k] G_i[k] A^{-1}[k] w[k]\\
		&+K_i[k+1] v_{i}[k+1],
	\end{align*}
	where the last equality comes from \eqref{eq:defG}.
The proof is completed. \hfill \QED

\section{Proof of Lemma~\ref{lm:sumW}}
\label{lmpf:sumW}
According to \eqref{eq:defW}, we know that $\bm{W}_{i j}[k]$ satisfies:
	\begin{align*}
		\bm{W}_{i j}[k+1]= &\Pi[k] \bm{W}_{i j}[k] \Pi^{\top}[k]+
		\\
		&\Pi[k]G_i[k]A^{-1}[k] Q[k] \left(\Pi[k]G_j[k]A^{-1}[k]\right)^{\top}+
		\\
		&K_i[k+1]K_j^{\top}[k+1] \circ \left( R_{i j}[k+1]\otimes \mathbf{1}_{n\times n} \right),
	\end{align*}
	where scalar $R_{i j}[k+1]$ is the element of the matrix $R[k+1]$ on $i$-th row and $j$-th column. On the other hand, since $\sum_{i=0}^{m}G_i[k]=I$ is shown in Lemma~\ref{lm:sumG}, one finds that
	\begin{align}
		&\sum_{i=1}^{m} \bm{W}_{i j}[k+1] \notag \\
		=&\Pi[k] \left(\sum_{i=1}^{m} \bm{W}_{i j}[k]\right) \Pi^{\top}[k] \notag\\
		& +\left(I-K[k+1]C\right)Q[k]\left(\Pi[k]G_j[k]A^{-1}[k]\right)^{\top} \notag\\
		& + K[k+1]R_j[k+1]K_j^{\top}[k+1]. \label{eq:sumWij}
	\end{align}
	In the following, we prove that $P[k]G_j^{\top}[k]$ satisfies the same dynamics with $\sum_{i=1}^{m}  \bm{W}_{i j}[k]$, where $P[k]$ is defined in \eqref{eq:def_Pt_asy}.
	According to \eqref{eq:asy_kalman}, $P[k]$ and $K[k]$ satisfy the following:
	\begin{align}
		&P[k+1]=\left(I-K[k+1]C\right) \left(A[k]P[k]A^{\top}[k] +Q[k]\right), \notag \\
		&K[k+1]R[k+1]=\Pi[k]P[k]A^{\top}[k]C^{\top}[k+1] \notag \\
		&\hspace{80pt}+\left(I-K[k+1]C\right)Q[k]C^{\top}[k+1]. \label{eq:KR}
	\end{align}
	
	Considering the dynamics of $G_j[k]$ in \eqref{eq:defG} gives us the following:
	\begin{align}
		&P[k+1]G_j^{\top}[k+1] \notag \\
		=&\Pi[k]P[k]A^{\top}[k] G_j^{\top}[k+1]
		+\left(I-K[k+1]C\right)Q[k]G_j^{\top}[k+1] \notag \\
		=&\Pi[k]P[k] G_j^{\top}[k] \Pi^{\top}[k]+\Pi[k]P[k]A^{\top}[k] C_j^{\top}[k+1] K_j^{\top}[k+1] \notag \\
		&+\left(I-K[k+1]C\right)Q[k]G_j^{\top}[k+1] \notag \\
		=&\Pi[k]P[k] G_j^{\top}[k] \Pi^{\top}[k] \notag \\
		&+\left(I-K[k+1]C\right)Q[k]\left(\Pi[k] G_j[k] A^{-1}[k]\right)^{\top} \notag \\
		&+\left(\Pi[k]P[k]A^{\top}[k]+\left(I-K[k+1]C\right)Q[k]\right)\times \notag \\
		&\hspace{120pt} C_j^{\top}[k+1] K_j^{\top}[k+1]. \label{eq:PG}
	\end{align}
    From \eqref{eq:sumWij}-\eqref{eq:PG}, one obtains 
    $\sum_{j=1}^{m}  \bm{W}_{i j}[k]=P[k] G_i^{\top}[k]$.
\hfill \QED

\section{The construction of $W[0]$}
\label{contruction_W0}
Construct 
		\begin{align*}
				\bm{W}[0]\triangleq
				\bm{D}
				\circ
				\left(\mathbf{1}_{m\times m}\otimes \Sigma\right) ,
			\end{align*}
		where 
		\begin{align*}
				\bm{D}\triangleq
				\begin{bmatrix}
						\bm{D}_{11} & \bm{D}\bm{D}_{12} & \cdots & \bm{D}_{1m}\\
						\bm{D}_{21} & \bm{D}_{22} & \cdots & \bm{D}_{1m}\\
						\vdots & \vdots & \ddots & \vdots\\
						\bm{D}_{m1} & \bm{D}_{m2} & \cdots & \bm{D}_{mm}\\
					\end{bmatrix}.
			\end{align*}
		One can verify that, if the following three constraints are satisfied, $\bm{W}[0]$ is Hermitian, strictly positive definite and satisfies $\sum_{j=1}^{m}  \bm{W}_{ij}[0]=\Sigma\cdot G_i^{\top}[0]$.
		\begin{enumerate}
				\item $\bm{D}=\bm{D}^\top$,
				\item $\bm{D}\succ 0$,
				\item $\sum_{j=1}^{m} \bm{D}_{ij}= G_i[0]$ for all $i\in\Ic$.
			\end{enumerate}
            
		We design the blocks $\bm{D}_{ij}$ to be the following diagonal matrices:
		\begin{align} \label{def:Dij}
				\bm{D}_{ij}\triangleq \begin{cases}
						-I_n, \text{ if } i\neq j, \\
						G_i[0]+(m-1)\cdot I_n, \text{ if } i= j,
					\end{cases}
			\end{align}
		where $I_n$ represents an $n\times n$ identity matrix.

        By definition \eqref{def:Dij},
		the conditions (1) and (3) are satisfied. We proceed to prove that $\bm{D}$ is positive definite.	
		Denote $\bm{D}^{[k]}_{ij}$ as the $k$-th diagonal element of $\bm{D}_{ij}$.	
		We have that $\bm{D}^{[k]}_{ii}\geq\sum_{j\neq i}\left| \bm{D}^{[k]}_{ij} \right|$ for all $k\in\Jc$ since $G_i[0]$ is non-negative diagonal matrix.
		According to Gershgorin circle theorem, $\bm{D}$ is positive semi-definite. 
		
		We proceed to prove that $\bm{D}$ is positive definite after elementary matrix operation. 
		Since the system is observable, for each state index $j\in\Jc$, there exists a sensor $i$ such that $i\in\Ec_j$. Denote such index $i$ as $\iota(j)$.
		For each $j\in\Jc$, we do the following examination procedure.
		For all $i\in\Ic$, if $i\notin \Ec_j$, then multiply the $1/\bm{D}^{(j)}_{\iota(j)j}$ times of $(\iota(j)-1)n+j$-th row of $\bm{D}$ on $(i-1)n+j$-th row of $\bm{D}$. 
		After this elementary matrix operation, one can verify that every diagonal matrix of $\bm{D}$ satisfies $\bm{D}^{[k]}_{ii}>\sum_{j\neq i}\left| \bm{D}^{[k]}_{ij} \right|$ for all $k\in\Jc$. Therefore, $\bm{D}$ is positive definite according to the fact that matrix rank does not degrade after an elementary matrix operation. \hfill \QED

\begin{lemma}
\label{lem:error_bound_med}
    Given a scalar variable $x$, let us consider the following function:
    \begin{align}
        f(x) = \sum_{i = 1}^a \, \lvert x + m_i \rvert + \sum_{j = 1}^b \, \lvert x + n_j \rvert,
    \end{align}
    where $a$ and $b$ are given positive integers; $m_i$ and $n_j$ are given real numbers for all $1 \leq i \leq a$ and $1 \leq j \leq b$.
    We denote $c \triangleq \left \lceil{\frac{b-a}{2}}\right \rceil$ and $n \triangleq \big[ n_1, \, n_2, \, \ldots, n_b  \big]^\top$.
    Suppose that the minimum value of $f(x)$ occurs at the optimal solution $x^\star$. If $b \geq a + 1$, then the optimal solution $x^\star$ fulfills the following:
    \begin{align}
        \lvert x^\star \rvert \leq \max \big \{ \lvert h_c(n) \rvert, \, \lvert h_c(-n) \rvert  \big \}, \label{xstarbound}
    \end{align}
    where $\max(y,\,z)$ takes the greater value between $y$ and $z$.
\end{lemma}
\begin{proof}
    Let us denote the vector $m \triangleq \big[ m_1,\,m_2,\ldots,\,m_a \big]^\top$ and $d \triangleq \left \lceil{\frac{a+b}{2}}\right \rceil$. The solution $x^\star$ to the minimum $f(x^\star)$ is any value in the following range: 
    \[
    r_d \triangleq \big[ -h_d\big( - \big[ m^\top,\,n^\top \big]^\top \big) , \,  h_d\big( \big[ m^\top,\,n^\top \big]^\top \big) \big].
    \]
    It is sufficient to consider two extreme cases, which are 1) the minimum value of $m$ is greater than the maximum value of $n$ and 2) the maximum value of $m$ is smaller than the minimum value of $n$. In the first case, one has $h_c(n) = h_d\big( \big[ m^\top,\,n^\top \big]^\top \big)$. Meanwhile, one has $-h_c(-n) = -h_d\big( - \big[ m^\top,\,n^\top \big]^\top \big)$. Therefore, one obtains $r_d \subset r_c \triangleq \big[ -h_c(-n),\,h_c(n) \big]$ for the second case. Those two observations result in $x^\star \in r_c$ and the upper bound of $x^\star$ represented in \eqref{xstarbound}. 
\end{proof}

\end{document}